\documentclass[11pt]{article}%
\pdfoutput=1
\usepackage[hidelinks]{hyperref}
\usepackage{bbm,url,microtype,amssymb,amsthm,mathtools,mathrsfs,cleveref,graphicx,float,color,subcaption}
\usepackage{fullpage}
\usepackage{lmodern}
\usepackage[T1]{fontenc}
\usepackage[USenglish]{babel}
\usepackage{authblk}
\usepackage[export]{adjustbox}

\providecommand{\U}[1]{\protect\rule{.1in}{.1in}}
\newtheorem{thm}{Theorem}\crefname{thm}{Theorem}{Theorems}
\newtheorem{lem}[thm]{Lemma}\crefname{lem}{Lemma}{Lemmas}
\newtheorem{prp}[thm]{Proposition}\crefname{prp}{Proposition}{Propositions}
\newtheorem{cor}[thm]{Corollary}\crefname{cor}{Corollary}{Corollaries}
\crefname{prb}{Problem}{Problems}
\newtheorem{dfn}[thm]{Definition}\crefname{dfn}{Definition}{Definitions}
\newtheorem{fac}[thm]{Fact}\crefname{fac}{Fact}{Facts}
\crefname{section}{Section}{Sections}
\crefname{appendix}{Appendix}{Appendices}
\numberwithin{equation}{section}
\let\oldref\ref
\renewcommand{\ref}[1]{(\oldref{#1})}
\DeclareMathOperator{\tr}{tr}

\DeclareMathOperator{\sgn}{sgn}

\begin{document}

\title{Coordinating Decisions via Quantum Telepathy}
\author[1]{Dawei Ding\thanks{dding@alumni.stanford.edu}}
\affil[1]{Independent}
\author[2]{Liang Jiang}
\affil[2]{Pritzker School of Molecular Engineering, University of Chicago, Chicago, Illinois 60637}
\date{}
\maketitle

\begin{abstract}
Quantum telepathy is the phenomenon where two non-communicating parties can exhibit correlated behaviors that are impossible to achieve using classical resources. This is also known as Bell inequality violation and is made possible by quantum entanglement. In this work, we present a conceptual framework for applying quantum telepathy to real-world problems. In general, the problems involve multiple parties making local observations that need to coordinate their decisions but are unable to communicate. 
We argue this inability is actually quite prevalent in the modern era where the decision-making timescales of computer processors are so short that the speed of light delay is appreciable in comparison. We highlight the example of high-frequency trading (HFT), where trades are made at microsecond timescales, but the speed of light delay between different stock exchanges are on the order of 100 microseconds to 10 milliseconds. Due to the maturity of Bell inequality violation experiments, experimental realization of quantum telepathy schemes that can attain a quantum advantage for real-world problems \emph{is already almost immediately possible}.
We demonstrate this by conducting a case study for a concrete HFT scenario that gives rise to a generalization of the CHSH game and evaluate different possible physical implementations for achieving a quantum advantage. It is well known that Bell inequality violation is a rigorous mathematical proof of a quantum advantage over any classical strategy and does not need any complexity-theoretic assumptions such as $\text{BQP}\neq\text{BPP}$. Moreover, fault tolerance is not necessary to realize a quantum advantage: for example, violating the CHSH inequality only requires single-qubit operations on two entangled physical qubits.
\end{abstract}

\tableofcontents

\section{Introduction}
Quantum entanglement is a unique testament to the strangeness of quantum mechanics. As stated in the celebrated result of John Bell~\cite{bell1964einstein}, no local hidden variable theory can reproduce the correlated behaviors of entangled particles. His prediction, which is known as Bell inequality violation, and its subsequent experimental verification \cite{freedman1972experimental,fry1976experimental,aspect1981experimental,aspect1982experimental}, completely overturned traditional assumptions of how nature works and made mathematically explicit the departure quantum mechanics makes from the familiar classical world. This phenomenon was also described by Einstein, Podolsky, and Rosen as ``spooky action at a distance''~\cite{einstein1935can}: although they are not communicating, entangled particles can display behaviors that seem inconceivable and even uncanny. 

Bell inequality violation is not a just topic of fundamental physics or philosophy. A series of works~\cite{brandenburger2016team,brandenburger2015quantum,hasanpour2017quantum,szegedy2020systems,szegedy2023systems} pointed out that \emph{Bell inequality violation has real-world applications.} 
To better elucidate how Bell inequalities can appear in real-world settings, the patents~\cite{szegedy2020systems,szegedy2023systems} propose a new concept called ``coordinating decisions between non-communicating parties'' (CDNP) problems, which we will henceforth refer to more simply as \textbf{tacit coordination} (TC) problems. We also introduce this terminology because we will make some generalizations and slight deviations from the usual theory of Bell inequalities and the equivalent computer science concept of nonlocal games~\cite{cabello2008necessary,cleve2004consequences}. In short, a TC problem arises when there are multiple parties who each make a local observation and a local decision. The aim is to optimize a global ``utility'' of their collective decisions given their collective observations, but \emph{without communicating during the decision-making process}.\footnote{See~\Cref{sec:nonlocal_games} for a more detailed description and see~\Cref{app:tc} for a technical definition. } Note that for consistency, we will use this terminology throughout the paper.
One very natural real-world setting where TC problems can appear is \emph{high-frequency trading} (HFT). In this setting, each party is an HFT server colocated at different stock exchanges. The servers are owned by the same trader. Each server has access to local information, such as stock price fluctuations at their respective exchange. Now, these stock exchanges are spatially separated by distances ranging from tens of kilometers to thousands of kilometers. The speed of light delay thus ranges from hundreds of microseconds to tens of milliseconds. However, modern HFT is conducted on timescales as short as microseconds and in the future may even shorten to nanoseconds~\cite{haldane2010}. No matter how much effort is put into shortening latencies, with these numbers it is physically impossible for different servers to communicate before having to make a trade decision. However, there could be a nontrivial globally optimal pair of trades given a pair of observations to minimize risk or maximize expected returns. See~\Cref{sec:latency_tc} for an example. Thus, HFT naturally gives rise to TC problems. 

Now, we emphasize Bell inequality violation \emph{is irrefutable proof of a quantum advantage}.\footnote{Quantum advantage means achieving a task with quantum resources that is difficult or even impossible to achieve with classical resources. For most of this paper, we will use this term to refer to a particular type of quantum advantage for TC problems. An explicit definition is given in~\Cref{sec:nonlocal_games}. } It does not require complexity-theoretic assumptions such as $\mathrm{BQP} \neq \mathrm{BPP}$ and is actually a straightforward mathematical argument. In fact, Bell inequality violation is so convincing that it can be used to prove a quantum advantage in other settings, such as shallow quantum circuits~\cite{bravyi2018quantum}. Furthermore, to achieve a quantum advantage in a TC problem using real hardware, we do not necessarily need fault-tolerant quantum computing. For example, violating the CHSH inequality~\cite{clauser1969proposed} only requires two physical qubits, along with single-qubit gates and measurements. Indeed, the hardware and fidelities necessary to violate the inequality have already been achieved half a century ago~\cite{freedman1972experimental,fry1976experimental,aspect1981experimental,aspect1982experimental}. Even for Bell inequalities that require high-dimensional quantum systems to achieve the maximum violation, it is sufficient to find (if it exists) a low-dimensional scheme that also violates the inequality, just not maximally, to attain a quantum advantage. Due to this potential of Bell inequality violation to have real-world applications, we attempt to introduce a term for Bell inequality violation that describes an operational concept, akin to ``quantum computing'' or ``quantum communication.'' We will use the term \textbf{quantum telepathy}.\footnote{There is a related term that already exists in the literature, \textbf{pseudo-telepathy}~\cite{brassard2005quantum}. This is a special case of quantum telepathy where the parties can attain the algebraically largest possible value of the Bell expression with a quantum strategy, for example in the GHZ game or the magic square game. }

The main message of this paper is that \emph{TC problems are actually quite prevalent in the real world and that a quantum advantage can be attained with current or near-future technologies}. 
The core insight is that speed of light delays for distances on the order of 10 km is already 10 to 100 $\mu$s, 
which is a relatively long period of time in our modern world satiated with classical processors that have GHz clock speeds.
In particular, we conduct a case study of a concrete HFT scenario that gives rise to a generalization of the CHSH game and assess different possible physical implementations of quantum telepathy schemes.
Our paper will be organized as follows. In~\Cref{sec:nonlocal_games}, we first give an informal definition of TC problems. Then, in~\Cref{sec:latency_tc}, we present the concrete HFT scenario and compute what quantum advantages can be attained. In~\Cref{sec:physical}, we go into the detailed physical implementations of quantum telepathy schemes. In particular, we evaluate their ability to achieve a quantum advantage for the said HFT scenario, considering practical issues such as photon loss, entanglement generation rates, and robustness to noise. We end with a discussion in~\Cref{sec:discussion}, where we consider other settings where TC problems may arise such as distributed computing and computer architecture.~\Cref{app:tc} gives technical definitions and proves facts about TC problems that we use in our analyses. In particular, we give definitions of a TC problem in a more general setting than the one given in~\Cref{sec:nonlocal_games}, involving an arbitrary number of parties, inputs, and outputs. The contents are also presented in a more formal, mathematical style.
Furthermore, although the effects of loss discussed in~\Cref{app:loss} and of noise in~\Cref{app:noisy} are familiar concepts in the literature, to our knowledge, the explicit mathematical definitions and the results we prove for the general multipartite case have not been written down before.~\Cref{app:numerical} presents a numerical optimizer for computing the quantum advantage attainable for general TC problems by explicitly parameterizing a quantum strategy while reducing the number of parameters as much as possible. We are not aware of any previous work that does this systematically.
Finally,~\Cref{app:arch} gives a ``pseudo-example'' of a TC problem that can arise in a computer architecture setting.

\paragraph{Related work} Quantum telepathy is not a new concept. There are multiple papers discussing applications of quantum telepathy~\cite{brassard2005quantum,ashkenazi2022distributed,broadbent2008can,huberman2003quantum}, but none of these point out that the speed of light delay is a fundamental barrier for communication in many real-world settings. Furthermore, many of the examples given in these works are mathematically motivated and are not manifestly relevant for real-world problems. 

After this paper was posted on arXiv, the paper~\cite{brandenburger2016team} and the related patent~\cite{brandenburger2015quantum} were brought to our attention. To our knowledge, these were the first works to propose using quantum telepathy for HFT as well as to give a concrete trading scenario. 
However, the scenario involves two traders that are at a large distance from each other and from the two markets (New York and Shanghai) they want to trade at. This is inconsistent with the standard practice in HFT of using servers co-located in the stock exchange's data center~\cite{nasdaq}. We rectify this with the example in~\cref{sec:latency_tc}.
The patents~\cite{szegedy2020systems,szegedy2023systems} were filed without knowledge of the earlier works. The core problem in all of these previous works mentioning HFT is that the distance scales considered were tens of thousands of kilometers. Distributing entanglement across such a distance would require a quantum satellite~\cite{yin2017satellite}, a setup with very high rates of photon loss, among other practical problems. Furthermore, the previous works also do not go into much detail about physical implementations nor quantitatively evaluate their ability to achieve a quantum advantage. We note the ideas of this current paper were included in a new patent~\cite{ding2024system}. Another related work is an analysis of using quantum computers to solve Prisoner's Dilemma which is used to model HFT given in~\cite{khan2021quantum}. After our initial arXiv post, we also came upon the paper~\cite{hasanpour2017quantum}, which is an application of the CHSH inequality to ad hoc network routing.

\section{Definition of Tacit Coordination (TC) Problems}
\label{sec:nonlocal_games}
We give an informal definition of TC problems. For a technical exposition, see~\Cref{app:tc}. A TC problem involves multiple \textbf{parties} that each make an \textbf{observation} followed by a \textbf{decision} so as to maximize a global \textbf{utility}. In~\Cref{fig:two_party_game}, we give an illustration of a TC problem. For simplicity, we will assume there are two parties. 
\begin{figure}[htp]
    \centering
    \includegraphics[width=0.23\textwidth]{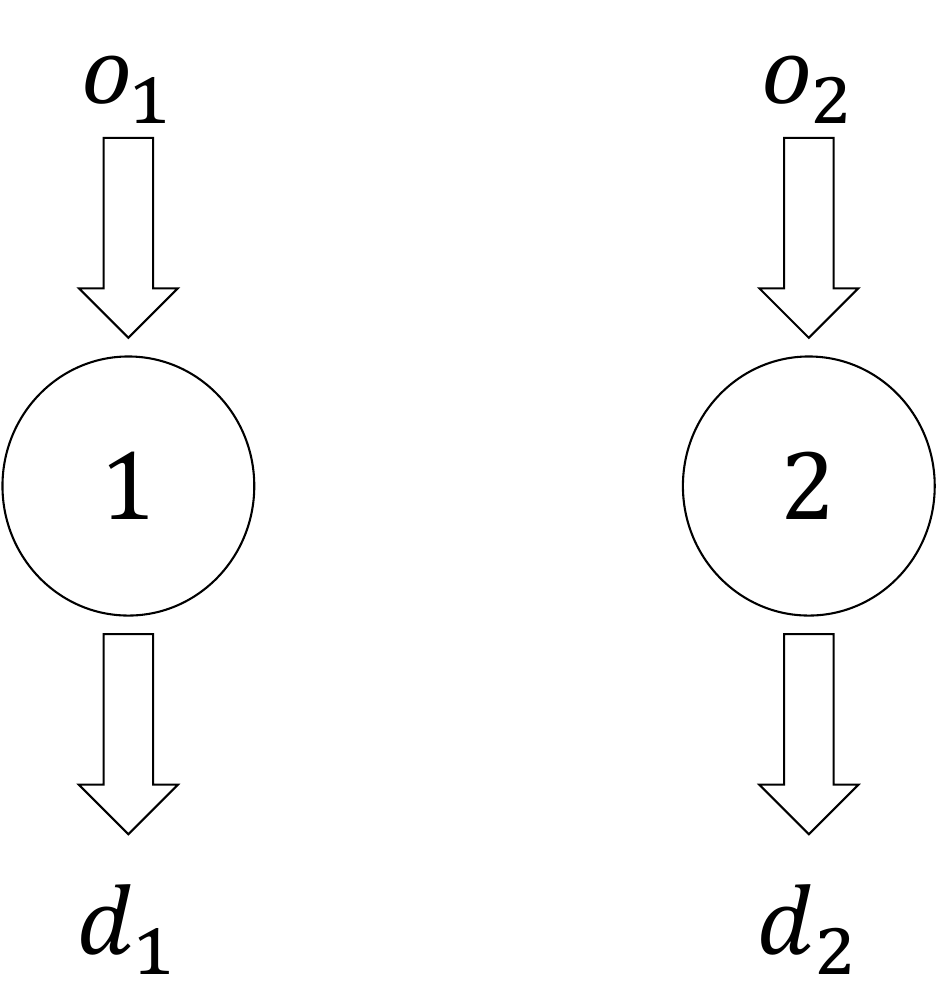}
    \caption{A TC problem with two parties. The first party makes observation $o_1$ and subsequent decision $d_1$ while the second party makes observation $o_2$ and subsequent decision $d_2$. }
    \label{fig:two_party_game}
\end{figure}
A \textbf{utility array} describes what is the utility of a set of decisions given a set of observations. For simplicity, suppose each party can only have two possible observations $\{o,o'\}$ and two possible decisions $\{d,d'\}$. The utility array can be written in matrix form:
\begin{align*}
    \mathcal U =
    \bordermatrix{
    & d,d & d,d'& d',d & d',d' \cr
    o,o & u_{11} & u_{12} & u_{13}& u_{14}\cr
    o,o' & u_{21}& u_{22}& u_{23}& u_{24}\cr
    o',o & u_{31}& u_{32}& u_{33}& u_{34}\cr
    o',o'& u_{41}& u_{42}& u_{43}& u_{44} \cr
    } \in \mathbb{R}^{4\times 4},
\end{align*}
where the rows are labeled by possible pairs of observations and the columns are labeled by possible pairs of decisions. The matrix element is then the utility of that pair of decisions given that pair of observations. The parties are to each make a decision given their observation that maximizes the expected value of the utility, given an \textbf{input distribution}, that is, a probability distribution over the observations. The catch is that the parties cannot communicate during this decision-making process. The reason for this could be various: 
\begin{itemize}
    \item The parties may be spatially separated and have to make decisions faster than the communication latency or even the speed of light delay.
    \item The parties may have lost their communication link due to a technical fault or other problem.
\end{itemize}
The former case can be quite prevalent in modern settings: take the parties to be computers that perform tasks with GHz clock speeds. As long as the spatial separation between the computers is more than 30 cm, communication is physically impossible.

Without communication, how do the parties maximize their utilities? If they're smart, the parties would have agreed on a \textbf{strategy} beforehand given knowledge of the utility array. The implementation of the strategy leads in general to a probability distribution of decisions conditioned on observations, which can also be expressed as a matrix:
\begin{align*}
    p(d\vert o) = 
    \bordermatrix{
    & d,d & d,d'& d',d & d',d' \cr
    o,o & p_{11} & p_{12} & p_{13}& p_{14}\cr
    o,o' & p_{21}& p_{22}& p_{23}& p_{24}\cr
    o',o & p_{31}& p_{32}& p_{33}& p_{34}\cr
    o',o'& p_{41}& p_{42}& p_{43}& p_{44} \cr
    } \in \mathbb{R}_{\geq 0}^{4\times 4}.
\end{align*}
Because these are probabilities, this is a right stochastic matrix, i.e.\ the rows sum to 1. We will call this conditional probability distribution the \textbf{behavior} of the parties. The expected utility is then given by multiplying these probabilities with the entries of the utility array and some input distribution $p_O(o)$.

In a \textbf{deterministic strategy}, each individual party simply makes a decision based on their local observation:
\begin{enumerate}
    \item Before starting, the parties prepare respective functions $f_1, f_2 : \{o,o'\} \to \{d,d'\}$ based on the utility matrix.
    \item After starting, the parties compute their respective functions by taking their respective local observation as input and use the corresponding output as their decision.
\end{enumerate}
In general, a \textbf{classical strategy}, also known as a local hidden variable theory, can involve randomness: both local randomness for each party and shared randomness establishing correlations between different parties. However, the expected utility of any classical strategy is always a convex combination of the expected utilities of deterministic strategies. Hence, to find the maximum possible expected utility function of all classical strategies, it is sufficient to only consider deterministic strategies. 

Now, the phenomenon of Bell inequality violation in the language of TC problems is this: for certain TC problems, \emph{we can devise strategies using quantum mechanics that achieve higher expected utilities than that of any strategy we can devise classically}. Such \textbf{quantum strategies} are of the following form:
\begin{enumerate}
    \item Before starting, the parties share entangled particles. The precise state of the entangled particles is specifically designed for the utility matrix.
    \item After starting, the parties apply a quantum measurement to their particle, the type of measurement being based on their observation.
    \item The parties make decisions based on the measurement result.
\end{enumerate}
Note here that no communication is conducted at any step. However, the parties need to be able to share entangled particles. 

In general, we will be interested in the \textbf{quantum value} and \textbf{classical value}, which are defined as the expected utilities optimized over quantum and classical strategies, respectively. We will refer to the difference between the quantum and classical value as the \textbf{gap}, or \textbf{quantum advantage}. For detailed technical definitions of the concepts in this section, readers can refer to~\Cref{app:tc}. 


\section{High Frequency Trading: A Case Study}
\label{sec:latency_tc}
We currently live in a world where the speed of light delay between different parties can be appreciable compared to the timescales in which decisions need to be made. Classical processors have GHz clock speeds, while light in vacuum can only travel 30 cm in 1 ns. Hence, TC problems can easily arise in real-world scenarios where communication is not possible due to latency. In general, we give the following criteria for a \emph{latency TC problem}:
\begin{enumerate}
    \item Multiple parties are involved that each make local observations and decisions.
    \item There is a global utility associated with a set of decisions given a set of observations.
    \item The parties do not have enough time to communicate with each other before having to make a decision.
\end{enumerate}
The last criterion could be due to fundamental physics constraints such as speed of light (in vacuum) delay or other factors such as having only an optical fiber connection through which light has to travel farther than the displacement between the parties and travels more slowly than in vacuum. We mention that there are other kinds of TC problems where the parties cannot communicate for reasons other than latency, for example the setting in~\cite{hasanpour2017quantum}.

In this paper we will conduct a detailed case study involving a high frequency trading scenario where quantum entanglement can provide an advantage. We will attempt to make the scenario as detailed as possible, but some simplifications will be necessary due to lack of real market data and to keep things conceptually straightforward. The TC problem we present should therefore be interpreted as a toy model. In general, to obtain concrete quantitative predictions of quantum advantage in practical settings, real historical HFT data is needed.  We hope that this HFT scenario and the quantum advantage that can be attained will inspire other interesting examples and attract further research in this area.

\begin{figure}[h!]
    \centering
    \includegraphics[width=0.8\textwidth]{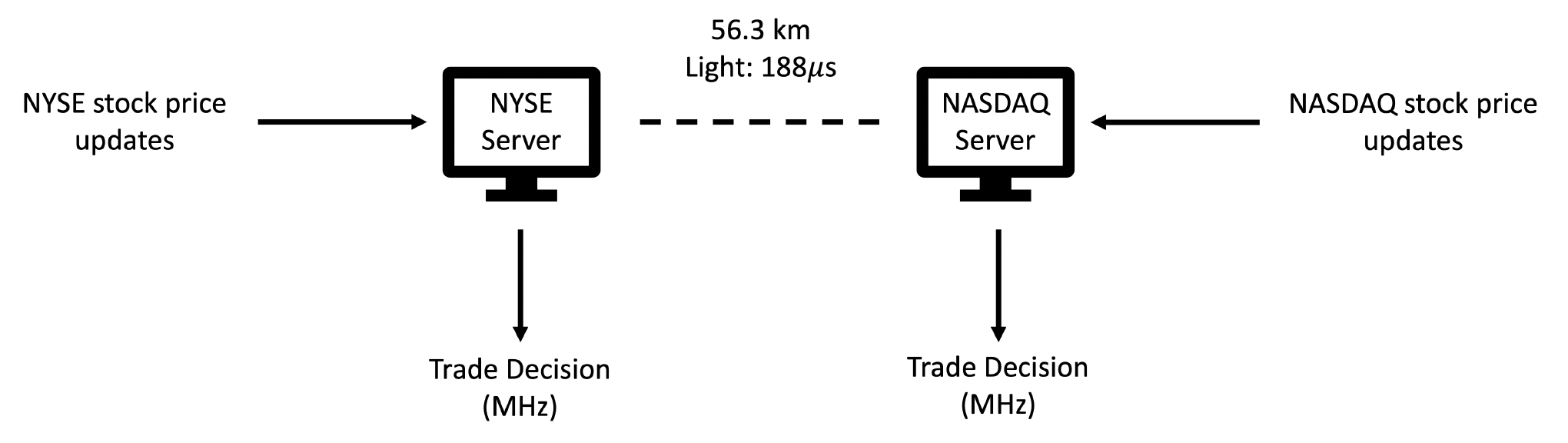}
    \caption{HFT setup between NYSE and NASDAQ exchanges. }
    \label{fig:nyse_nasdaq}
\end{figure}
Consider the HFT setup in~\Cref{fig:nyse_nasdaq}. Suppose a market maker named Zhuo operates in NYSE and NASDAQ, where he respectively trades correlated stocks X and Y. In layman's terms, Zhuo is a supplier: he provides trades (both buying and selling) for a stock X listed in NYSE and a stock Y listed in NASDAQ near current market prices so that trading can proceed. Zhuo has two colocated servers, one at each exchange. Each server receives market data from its local exchange (observations) and modifies its HFT algorithm accordingly (decision). Thus the first criterion for latency TC problems is satisfied. Since the stocks are correlated, the pair of decisions given the observations will determine how favorable Zhuo's position is given the market information at both exchanges (utility). Thus the second criterion is fulfilled. Furthermore, the data centers of these two stock exchanges are around 35 miles, or 56.3 km, apart~\cite{nyse_nasdaq} which means even a line of sight connection in vacuum would entail a speed of light delay of at least about 188 $\mu$s. However, HFT is sometimes conducted on timescales as short as microseconds~\cite{haldane2010}, which means in general the two colocated servers do not have enough time to communicate before they have to make a trade decision. Thus the third criterion is satisfied.

We now provide the full details of this particular TC problem. As Zhuo is a market maker, each server conventionally issues orders in pairs where one order is a bid (buy) and the other an ask (sell). However, the pair of orders is issued sequentially,\footnote{In general this will depend on the order types available at the exchanges. For example, some order types send both the ask and bid orders in one message, but even then the server has to sequentially decide on their respective price values. A possible exception to this is if the order type allows for specifying a spread. } and the first order is in general more likely to be filled. This is due to the first order arriving earlier and having being made with slightly more up-to-date information. Thus, if a bid order is always issued first in the pair, then there will be a slight bias toward successful bid orders over ask orders.\footnote{The significance of the bias due to which order is issued first may depend on the specific trading setting. Another way to increase the bias is to create a very slight price differential (that is, increase price for bids, decrease for asks). Note in this case we would need to evaluate the tradeoff between the loss in revenue and the reduction in risk. } For this TC problem, which order to issue first is the decision to be made for each server. 

Now, assume the price movements of stock X and Y are usually positively correlated. In this case, Zhuo should let the servers make opposite decisions in order to \textbf{hedge}, that is, to reduce risk, commonly measured by the variance of the portfolio value. That is, if one server decides to ask first, the other should bid first. This way Zhuo's portfolio will effectively have a nice tradeoff between two positively correlated stocks, thereby reducing the variance of his position. However, each server will also look for a technical indicator that stock X and Y are now negatively correlated. The precise nature of such an indicator would need to be empirically determined in real-world settings. For this TC problem, this indicator is the observation that each server makes. 
In this case, Zhuo should let the two servers make the same decision, again for hedging.
Now, both exchanges are looking for this indicator. If both see the indicator, then Zhuo should definitely let the servers make the same decision. If only one sees an indicator, there can be slight preference for making the same decision over opposite decisions or vice versa.

At a high level, we see that in this HFT scenario the utility (favorability of Zhuo's position) is only dependent on whether the servers make the same decision or different decisions. Hence, it is an instance of an XOR game~\cite{cleve2004consequences}. Here, the XOR of the decisions of the two servers has a natural trading interpretation as hedging for positively or negatively correlated stocks. This insight could be generalized to other HFT scenarios where what matters is the relationship between the trade decisions of the different servers and not the individual decisions themselves. Such scenarios could be mapped to a unique game~\cite{rao2008parallel} or a $d$-to-$d$ game~\cite{khot2002power}. We leave this to future work.
We next wish to explicitly write out the utility array. Assuming that only one indicator of potential negative correlation is not sufficient to prefer making the same decision over opposite decisions, the utility array is actually that of the opposite winning conditions of the CHSH game:
\begin{align*}
    \mathcal U =
    \bordermatrix{
    & A,A & A,B& B,A & B,B \cr
    N,N & 0 & 1 & 1 & 0 \cr
    N,I & 0 & 1 & 1 & 0 \cr
    I,N & 0 & 1 & 1 & 0 \cr
    I,I & 1 & 0 & 0 & 1
    },
\end{align*}
where the observations $N,I$ correspond to ``No indicator'' versus ``Indicator'' of negative correlation, and the decisions $A,B$ correspond to ``Ask-first'' or ``Bid-first'' order issuing, respectively. We will choose the order of parties as the NYSE server followed by the NASDAQ server. \emph{This establishes that a TC problem such as the CHSH game can manifest in real-world scenarios}.\footnote{For simplicity, we assume that the figure of merit is indeed the expected utility. This is sensible enough for a toy model. Evaluating different figures of merit in detail would require working with real market data. } Now, to quantitatively evaluate a quantum advantage, we also need to consider the input distribution. Suppose the input distribution is independent Bernoulli distributed with parameter $p$ for both servers, where with probability $p$ the indicator is observed (I). We analytically solve this TC problem in~\Cref{app:qadv} where we find that a quantum advantage exists iff $p \in (1 - \frac{1}{\sqrt{2}}, \frac{1}{\sqrt{2}})$.

We can also consider a natural extension where one server observing the indicator leads to a partial preference of making the same over opposite decisions. In this case we can introduce a parameter $\beta \in [0,1]$ and the utility array is given by
\begin{align}
\label{eq:hon_utility}
    \mathcal U =
    \bordermatrix{
    & A,A & A,B& B,A & B,B \cr
    N,N & 0 & 1 & 1 & 0 \cr
    N,I & \beta & 1-\beta & 1-\beta & \beta \cr
    I,N & \beta & 1-\beta & 1-\beta & \beta \cr
    I,I & 1 & 0 & 0 & 1
    }.
\end{align}
Assuming again an independent Bernoulli distribution for the input distribution with parameter $p$, we can numerically evaluate the quantum advantage. The result is shown in~\Cref{fig:anti_toxic} for different values of $p, \beta$, as well as different cross-sections. These and other numerical results are obtained via a combination of the standard semidefinite programming algorithm for XOR games~\cite{cleve2004consequences} and a general-purpose numerical optimizer applicable to any TC problem outlined in~\Cref{app:numerical}.
\begin{figure}
    \centering
    \begin{subfigure}[b]{0.85\textwidth}
    \centering
    \includegraphics[width = 0.95\textwidth, right]{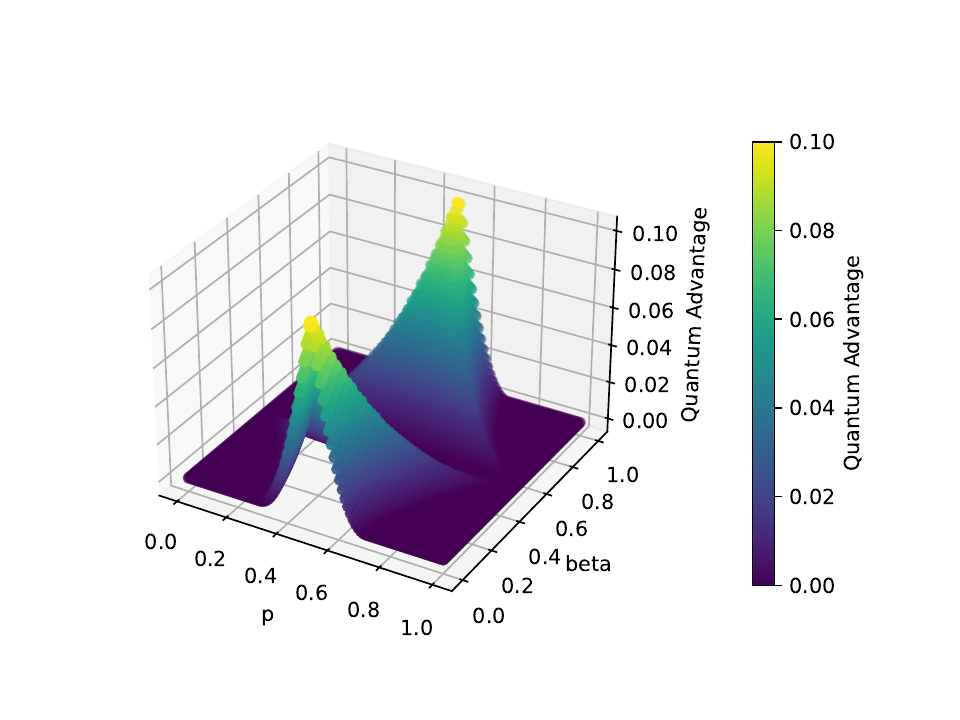}
    \caption{}
    \end{subfigure}
    \\
    \begin{subfigure}[b]{0.48\textwidth}
    \centering
    \includegraphics[width = 0.90\textwidth, left]{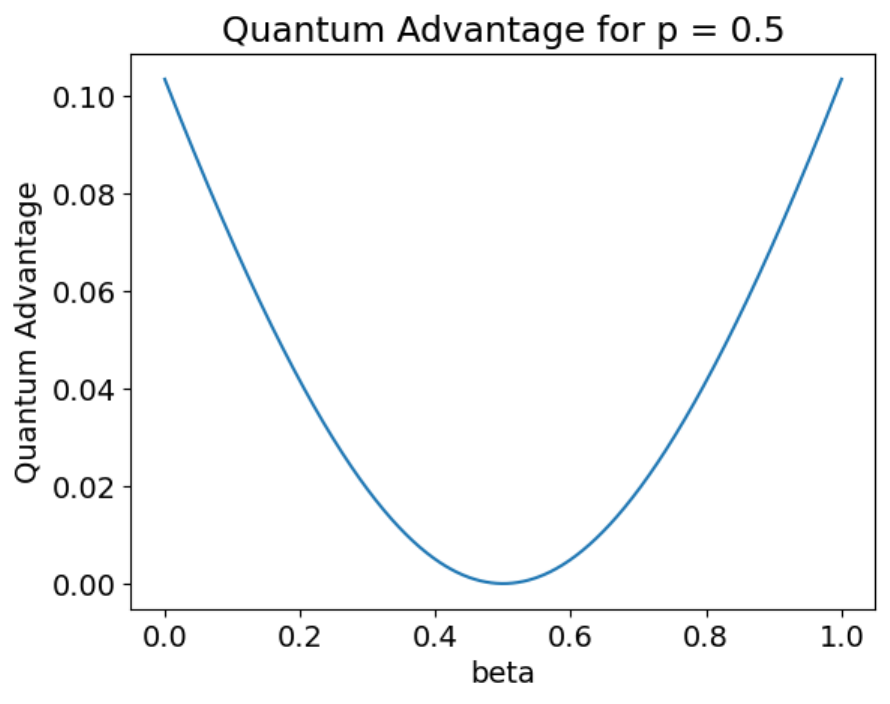}
    \caption{}
    \end{subfigure}
    \quad
    \begin{subfigure}[b]{0.48\textwidth}
    \centering
    \includegraphics[width = 0.9 \textwidth, left]{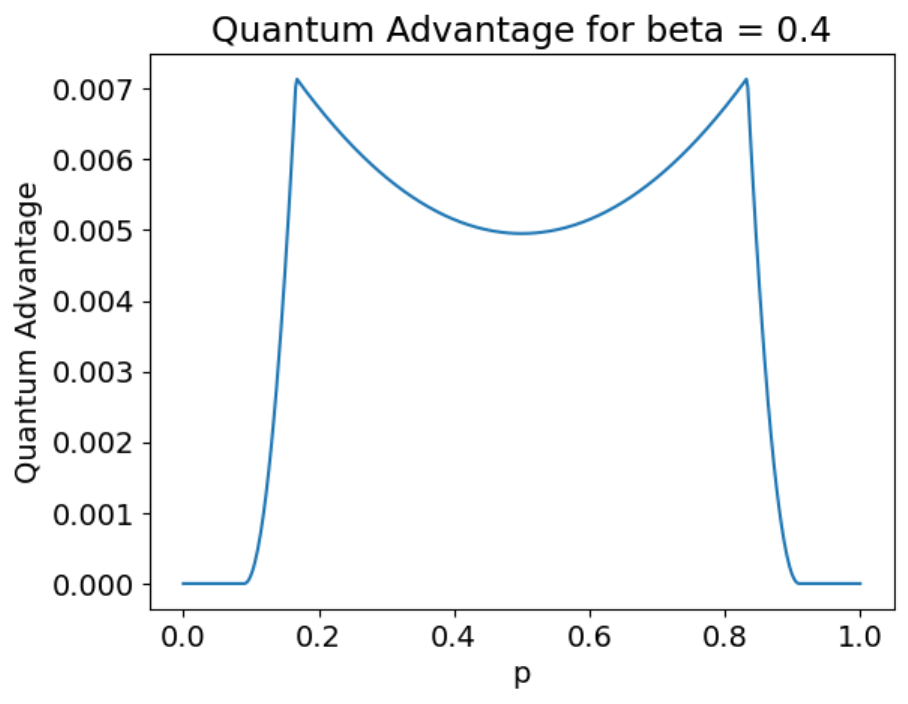}
    \caption{}
    \end{subfigure}
    \caption{(a) Quantum advantage for the hedging problem as a function of $p$ and $\beta$. 
    (b) Quantum advantage as a function of $\beta$ for $p = 0.5$. (c) Quantum advantage as a function of $p$ for $\beta = 0.4$. }
    \label{fig:anti_toxic}
\end{figure}
We see that when $\beta=0$, we obtain the previous results where there is a quantum advantage for $p \in (1-\frac{1}{\sqrt{2}}, \frac{1}{\sqrt{2}})$. Interestingly, as $\beta$ increases to 0.5, the range of $p$ for which there is a quantum advantage increases, although the value of the quantum advantage decreases. At $\beta =0.5$ there is no quantum advantage. For $\beta > 0.5$ we observe a mirror image by taking $\beta \mapsto 1-\beta$. This is because the corresponding utility array can be obtained by appropriately relabeling the observations and decisions, which will lead to the same gap as per~\Cref{lem:perm} in~\Cref{app:tc}. Hence, we see a family of real-world scenarios where quantum entanglement can provide an indubitable advantage. For convenience of reference, we will refer to this TC problem with the utility array given in~\Cref{eq:hon_utility} and with an independent Bernoulli distributed input with parameter $p$ as the \textbf{hedging problem}. The hedging problem can be easily generalized to multiple parties (more than 2) by considering more than 2 stocks, each at different exchanges, whose pairwise correlations can be positive or negative. Then, the decisions of the each corresponding colocated server have to match or or not match according to the observed correlations. The number of observations can also be increased by considering different indicators or other relevant local information. We leave further investigation of this generalized problem to future work. 

\section{Physical Implementation}
\label{sec:physical}
In this section we explore the concrete details of physically implementing a quantum strategy for a latency TC problem. We will in particular take the hedging problem as a case study. \emph{Our results suggest that there is a tantalizing possibility for quantum entanglement to bring practical gains to real-life TC problems such as those encountered in HFT with current or near-future quantum technologies}. 

In general, we will classify possible physical implementations into two types: direct photonic connection (Type I) and quantum memory (Type II). 

\subsection{Direct Photonic Connection}
\label{subsec:typeI}
In a Type I implementation, the two parties directly receive entangled photons from an intermediate source node. This is exhibited pictorially in~\Cref{fig:direct_photonic}.
\begin{figure}
    \centering
    \includegraphics[width=0.85\textwidth]{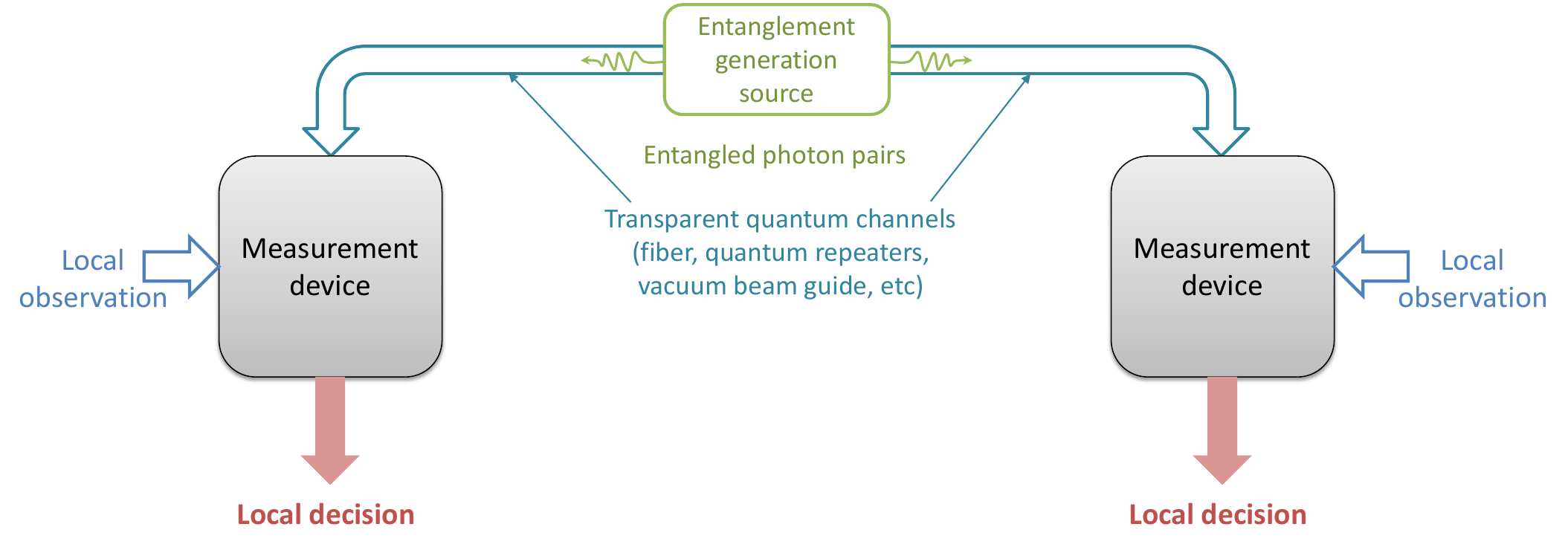}
    \caption{A physical implementation of a quantum strategy using a direct photonic connection (Type I). The two parties each have photon measurement devices that receive entangled photon pairs from an intermediate source. The measurement bases are determined by local observations and the local decisions are determined by the measurement outcomes. }
    \label{fig:direct_photonic}
\end{figure}

We will focus on two key specs for this setup to successfully implement a quantum strategy for a latency TC problem:
\begin{enumerate}
    \item Entangled photon generation rate
    \item Photon loss.
\end{enumerate}
A latency TC problem arises due to the extremely short time window in which decisions have to be made. Furthermore, for simplicity \emph{we will assume the same TC problem repeats itself}. This is sensible for many HFT scenarios: for instance, market makers such as Zhuo would be solving the same TC problem over and over again because their method of taking a favorable position or getting returns is always the same. The hedging problem is one such example.
For cases where the utility matrix changes with time, the parties and entangled photon source could synchronize changing the quantum strategy over time.
Nevertheless, this leads to the first requirement for the physical setup: the entangled photons are generated at a sufficiently high frequency so that a quantum strategy can be executed within the time window. In HFT, these time windows are on the order of microseconds, which means entangled photons need to be generated at MHz rates. This is already possible with current technologies using a continuous-wave laser and a nonlinear crystal to produce photon pairs maximally entangled in polarization via parametric down-conversion. These pairs can be produced at MHz rates (see for instance~\cite{yin2017satellite}). It is possible to accommodate even shorter time windows by using pumped pulses, thereby achieving generation rates up to $50$ GHz~\cite{wakui2020ultra}. We can also simply increase the number of entangled photon sources. These techniques may be necessary due to photon loss. However, note that it is \emph{not strictly necessary} that the entanglement generation rate be sufficiently high to attain a quantum advantage overall. The parties can simply execute the optimal deterministic strategy during times when entanglement is unavailable, and execute the optimal quantum strategy when it is.\footnote{Note that this would require the photons to be emitted at predetermined time intervals, that is, the source is deterministic. This was demonstrated in~\cite{muller2014demand}. } 

The second spec is photon loss. This is a key difficulty for a Type I implementation that could even preclude a quantum advantage~\cite{mermin1986new}. Because of photon loss, not only are photons sometimes unavailable, but  each party also does not know whether the other party received their photon.\footnote{There is also the problem of dark counts where one party erroneously detects a photon, which would constitute a different type of error. } If they both knew that one of the photons was lost, they could revert to the best classical strategy and thereby effectively achieve a nontrivial convex combination of the classical and quantum values. This would retain a quantum advantage. However, because they each do not know whether a photon was lost for the other party, they only have probabilistic control over what strategy they implement.  We explicitly define in~\Cref{app:loss} how exactly loss affects the resulting behavior of a quantum strategy. In general, the resulting behavior will be a convex combination of behaviors where the weight is the probability of photon loss and the parties that experience loss fall back to a deterministic strategy. As the probability of loss increases, the quantum advantage may eventually be lost.

Given a TC problem, assume each party successfully obtains his photon with probability $\eta$, which we call the \textbf{efficiency}. We can then define a \textbf{threshold efficiency} $\eta^*$ as the minimal value such that for $\eta > \eta^*$, we can attain a quantum advantage. If the TC problem had no quantum advantage even without loss, we define $\eta^* := 1$. Consider for example the hedging problem. We use the numerical optimizer outlined in~\Cref{app:numerical} and a modified NPA algorithm~\cite{navascues2008convergent,tba} to compute $\eta^*$ for various values of $p$ and $\beta$. We do this via a binary search on the highest value of $\eta$ where the classical value and the quantum value in the presence of loss are equal. Since the NPA algorithm yields upper bounds on the quantum value, this leads to lower bounds on $\eta^*$. Our numerical optimizer technically yields a lower bound on the quantum value, leading to an upper bound on $\eta^*$. We checked that there was reasonable agreement between the results from the modified NPA algorithm and those of the numerical optimizer for a subset of $p$ and $\beta$ values.
The results are shown in~\Cref{fig:eta_star}. 
\begin{figure}
        \centering
    \begin{subfigure}[b]{0.85\textwidth}
    \centering
    \includegraphics[width = 0.95\textwidth, right]{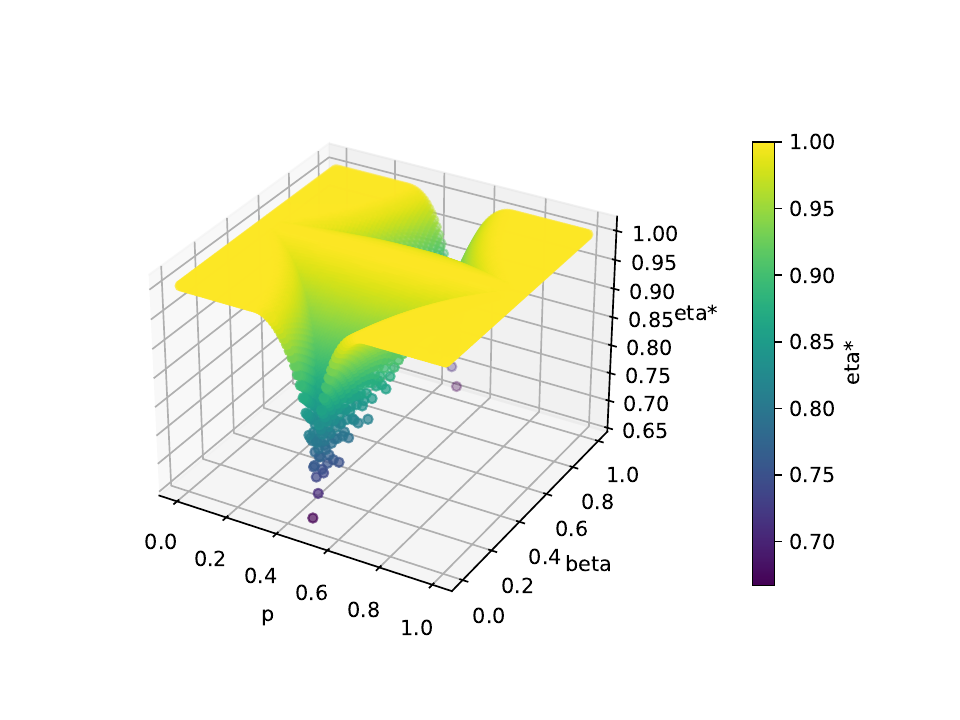}
    \caption{}
    \end{subfigure}
    \\
    \begin{subfigure}[b]{0.48\textwidth}
    \centering
    \includegraphics[width = 0.90\textwidth, left]{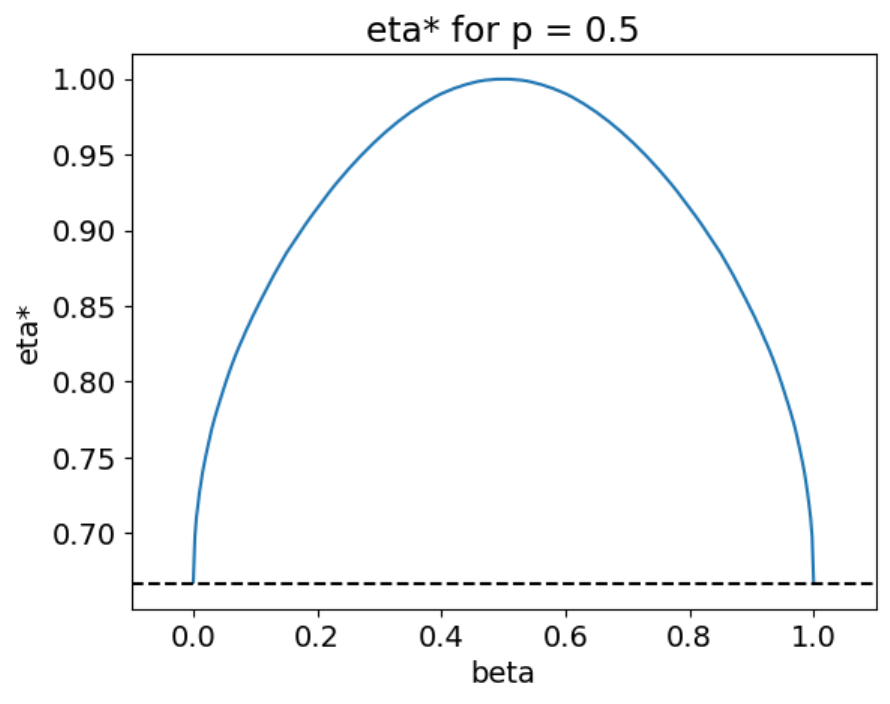}
    \caption{}
    \end{subfigure}
    \quad
    \begin{subfigure}[b]{0.48\textwidth}
    \centering
    \includegraphics[width = 0.9 \textwidth, left]{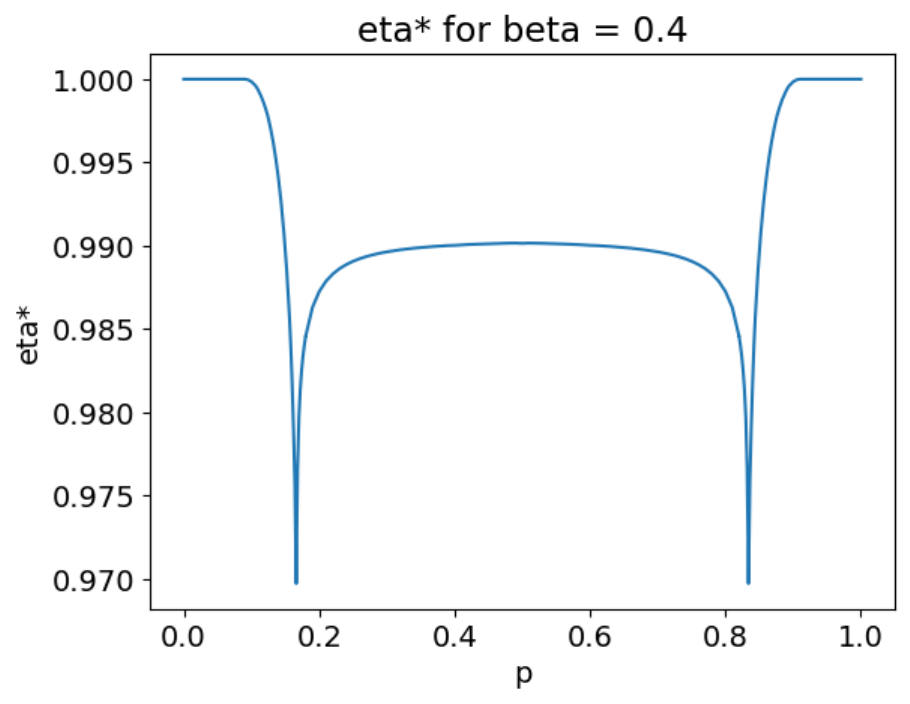}
    \caption{}
    \end{subfigure}
    \caption{(a) $\eta^*$ for the hedging problem as a function of $p$ and $\beta$ computed using the modified NPA algorithm. We only used a $101\times 101$ grid of points because of the long runtimes.
    (b) $\eta^*$ as a function of $\beta$ when $p = 0.5$ computed using our numerical optimizer. The lower bound of $\frac 2 3$ is plotted as a dashed line. (c) $\eta^*$ as a function of $p$ when $\beta =0.4$. }
    \label{fig:eta_star}
\end{figure}
In general, we expect that $1-\eta^*$ will behave similarly to the quantum advantage, since a larger advantage should yield a larger permissible loss rate. At a qualitative level, in~\Cref{fig:eta_star}(a), computed using our modified NPA algorithm, we see roughly the same pattern as in~\Cref{fig:anti_toxic}(a): $1-\eta^*$ values are larger for $p$ close to $0.5$, while $1-\eta^*$ values strictly greater than 0 is achievable for a wider range of $p$ for $\beta$ close to 0.5 (but not exactly). Note that we only use a $101\times 101$ grid of points (increments of 0.01) in~\Cref{fig:eta_star}(a) because computing $\eta^*$ involves computing quantum values at many different values of loss, leading to long runtimes.\footnote{We plotted~\Cref{fig:eta_star}(a) with the NPA algorithm because it is much faster than our numerical optimizer, which uses brute force search to guarantee the output is a global optimum. Even with this faster algorithm, the grid size that we can compute in a reasonable time is limited. } Cross-sections are plotted in~\Cref{fig:eta_star}(b) and (c), which were computed using our numerical optimizer.

It is known that for TC problems with two parties where each party has two possible observations, the threshold efficiency is at least $\frac 2 3$~\cite{massar2003violation}. This is also apparent in~\Cref{fig:eta_star}. If we use optical fiber to distribute the photons, this efficiency is already difficult to achieve for the distances involved for HFT. In general, $\eta$ exponentially decreases with the length of the channel $l$:
\begin{align*}
    \eta = 10^{-0.1 \alpha l},
\end{align*}
where $\alpha$ is the \textbf{attenuation rate} of the physical medium used for the channel. Now, the lowest attenuation rate achievable by optical fiber for the optimal wavelength of 1550 nm is 0.17 dB/km~\cite{cisco}. However, we compute for NYSE to NASDAQ
\begin{align*}
    \eta = 10^{-0.1 \cdot 0.17 \text{ dB/km} \times 56.3/2 \text{ km}} \approx 0.33 < \frac 2 3.
\end{align*}
Note that this is the distance from the source to each party is therefore half the distance between the two exchanges. To achieve an efficiency of at least $\frac 2 3$, the maximum length $l$ allowed is around 10.4 km. Furthermore, the efficiency computed should also include detector efficiency, inter-component coupling efficiency, and other possible sources of photon loss. This places an even stronger bound on the distance. In fact, we see $\eta^*$ values considerably higher than $\frac 2 3$ in~\Cref{fig:eta_star}, which shows that using optical fiber in a type I physical implementation is insufficient for the hedging problem when the two servers are separated by distances larger than $10.4$ km, such as in the case of NYSE to NASDAQ. 

There are other physical media to consider, however. The authors of~\cite{huang2023vacuum} propose using vacuum beam guides to distribute entangled photons. An attenuation rate as low as $5 \times 10^{-5}$ dB/km can be achieved according to numerical simulations. In this case, an efficiency of $\frac 2 3$ is achievable for distances less than about 35,000 km, which is about the circumference of the Earth. Thus this would allow for quantum strategies to be executed at \emph{continental distance scales}, for which speed of light delays can be significant. For example, the geodesic distance between the NYSE data center and the HKEX data centre is about $12.9 \times 10^3$ km, which is a delay of about 43.1 ms. This is very long compared to the time scale of HFT. For the case of NYSE and NASDAQ, using vacuum beam guides we can achieve $1-\eta$ of about $3.24 \times 10^{-4}$, which is sufficient for the majority of the points in~\Cref{fig:eta_star}. Note that optical fiber can still be useful for TC problems with shorter distances, or with a higher number of parties~\cite{larsson2001strict} or possible observations~\cite{massar2002nonlocality}.

For demonstration purposes, we next go into the fine details of strategies for the hedging problem, setting $p = 0.3$, $\beta =0.3$ for concreteness. In this case the threshold efficiency is computed to be $\eta^* \approx 0.941$. So that we attain a noticeable quantum advantage, we set $\eta = 0.95$. 
We run the numerical optimizer and find that the classical value is
\begin{align*}
    c^* = 0.79,
\end{align*}
where an optimal deterministic strategy is for the NYSE server to always bid first and the NASDAQ server to ask first when it sees no indicator and to bid first when it sees an indicator. The highest expected utility for a quantum strategy with efficiency $\eta = 0.95$ is
\begin{align*}
    q^*(\eta = 0.95) \approx 0.792.
\end{align*}
We prove in~\Cref{app:loss} that when there are only two possible observations and decisions for each party, qubit systems are sufficient for achieving the quantum value, even in the presence of loss. In~\Cref{app:numerical}, we show that we can parameterize each measurement basis with only one parameter. Define the states
\begin{align*}
    \vert \psi(\theta)\rangle & := \cos \theta \vert 0 \rangle  - \sin \theta \vert 1\rangle\\
    \vert \psi^\perp(\theta)\rangle & := \sin \theta \vert 0 \rangle  + \cos \theta \vert 1\rangle.
\end{align*}
One optimal quantum strategy uses the following measurement operators. When neither server sees an indicator, both servers measure in the computational basis\footnote{Since we only care about the largest eigenvalue of the Bell operator (see~\Cref{app:tc}), there is a local unitary degree of freedom and therefore we can without loss of generality always assume the measurement for the first observation is in the computational basis. }
\begin{align*}
    \{ \vert 0 \rangle\langle 0\vert, \vert 1 \rangle \langle 1 \vert \}.
\end{align*}
When either server sees an indicator, it instead uses the measurement 
\begin{align*}
    \{\vert \psi(-0.590) \rangle \langle \psi(-0.590) \vert, \vert \psi^\perp(-0.590) \rangle \langle \psi^\perp(-0.590) \vert\},
\end{align*}
where the angle is in radians and three significant figures are kept. The shared entangled state is
\begin{align*}
    \vert \Phi \rangle = 0.0401 \vert 00 \rangle - 0.902 \vert 01 \rangle - 0.428 \vert 10 \rangle - 0.0401 \vert 11\rangle,
\end{align*}
keeping three significant figures for the coefficients. 
To realize this quantum physically, we perform a Schmidt decomposition:
\begin{align}
\label{eq:schmidt}
    \vert \Phi \rangle = 0.903 \vert u_0 v_0 \rangle + 0.429 \vert u_1 v_1 \rangle,
\end{align}
where
\begin{align*}
    \vert u_0 \rangle := 0.9995 \vert 0 \rangle + 0.0301 \vert 1 \rangle, \vert u_1 \rangle := 0.0301 \vert 0 \rangle - 0.9995 \vert 1\rangle
\end{align*}
and
\begin{align*}
    \vert v_0 \rangle := 0.0301 \vert 0 \rangle + 0.9995 \vert 1 \rangle, \vert v_1 \rangle := -0.9995 \vert 0\rangle + 0.0301 \vert 1 \rangle.
\end{align*}
The state in~\Cref{eq:schmidt} can for example be realized by letting $\vert u_i \rangle, \vert v_i\rangle$ be the vertical-horizontal polarization basis and using the technique of~\cite{giustina2013bell}. The single-qubit rotations needed before the measurement to switch back to the computational basis can be realized using linear optics. Lastly, the optimal fallback deterministic strategy is for the NYSE server to always ask first and the NASDAQ server to always bid first. 



\subsection{Using Quantum Memory}
The main shortcoming of Type I implementations is that the magnitude of photon loss in optical fiber, the predominant physical medium for photonic communication in industrial applications, is too high for many TC problems. One solution for this is for the parties to each have a quantum memory, which could be a simple (few qubits) quantum computer with a long coherence time. The memories are entangled via photons in a heralded entanglement scheme~\cite{duan2001long}. This constitutes a Type II implementation. The setup is shown in~\Cref{fig:quantum_memory}.
 \begin{figure}
     \centering
     \includegraphics[width=0.85\textwidth]{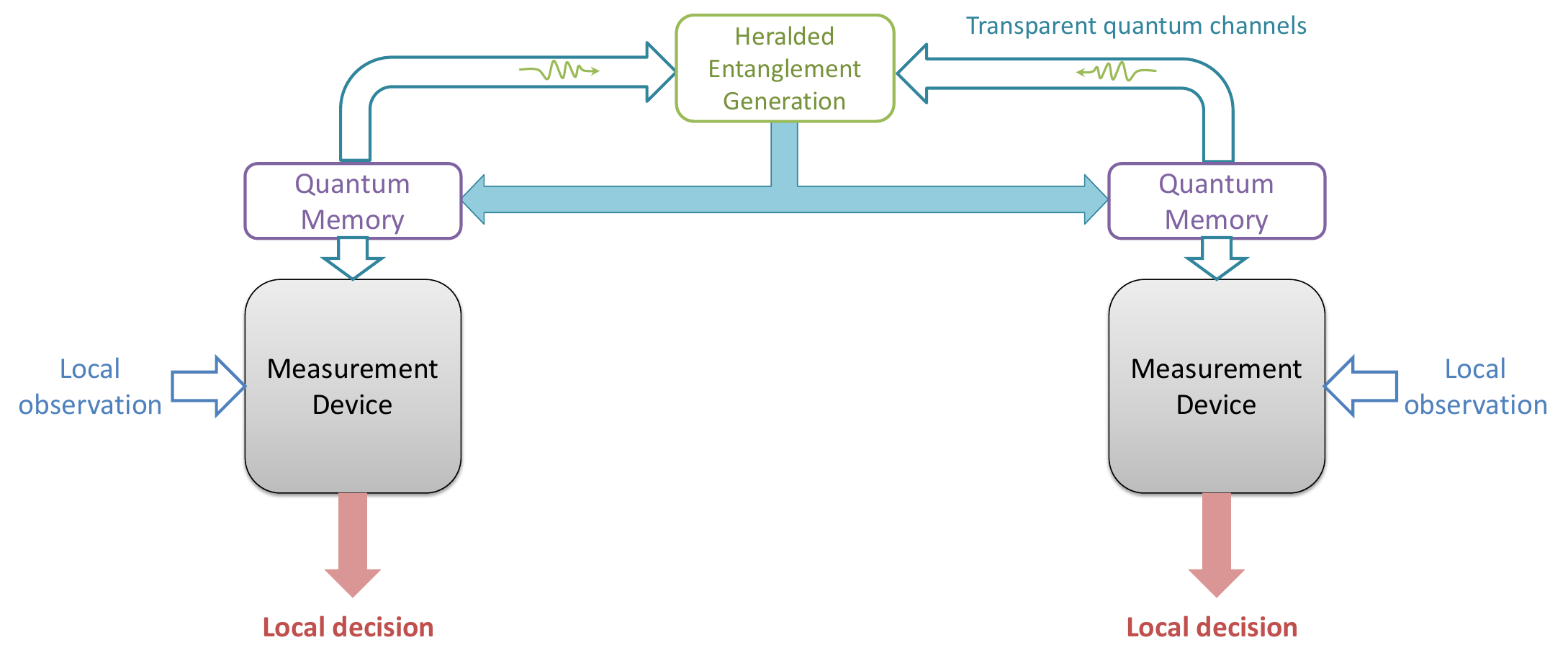}
     \caption{A physical implementation of a quantum strategy via quantum memories (Type II). Each party has a quantum memory, which are entangled via photons in a heralded entanglement scheme. The solid blue arrow indicates the heralding signal. The measurement bases are determined by local observations and the local decisions are determined by the measurement outcomes. }
     \label{fig:quantum_memory}
\end{figure}
Because the entanglement is heralded, we can avoid the problem in Type I implementations where the parties do not know if entanglement was successfully distributed. Furthermore, by using quantum memories we can realize more complicated and higher dimensional entangled states that may be difficult to realize with photons. 

The main specs we need to consider for Type II implementations are
\begin{enumerate}
    \item Effective entanglement generation rate $r_e$
    \item Fidelity of quantum operations $F$.
\end{enumerate}
The effective entanglement generation rate is determined by the following parameters:
\begin{itemize}
    \item Entanglement attempt time $t_a$
    \item Heralded entanglement success probability $p_s$
    \item Quantum memories multiplicity $M$
\end{itemize}
Here, $t_a$ is the time needed to generate entanglement for each quantum memory. This includes the time for the physical operation to create memory-photon entanglement, the time for the photons to travel to the intermediate node, and the time for the heralded entanglement information to be received. Usually $t_a$ is dominated by the traversal times. Next, $p_s$ includes the probability of successfully creating memory-photon entanglement, the probability of the photons reaching the intermediate node, and the probability of correctly projecting the photons into the desired state. Lastly, $M$ is the number of quantum memories available to each party. This increases the effective entanglement generation rate by a multiplicative factor. In summary, 
\begin{align*}
    r_e = \frac{M p_s}{t_a}.
\end{align*}
We would like a Type II implementation to have an effective generation rate sufficiently high for the TC problem's demand for entanglement. Again, like Type I implementations, this isn't necessary to get an overall quantum advantage. We simply implement the optimal deterministic strategy when entanglement is unavailable. Doing this is simple for Type II implementations since the entanglement is heralded.

We compute the effective entanglement generation rate for HFT between NYSE and NASDAQ as a function of $M$. Assuming $t_a$ is dominated by photon traversal time, we split it into two terms. The entangled photon traversal we will assume is through fiber with velocity $v_f$, while to minimize time for heralding we can use free space transmission with velocity $v_s$.\footnote{This is consistent with the use of microwave towers in HFT~\cite{arstechnica}. } Then, 
\begin{align*}
    t_a = \frac{56.3}{2} \text{ km} \div v_f + \frac{56.3}{2} \text{ km} \div v_s \approx 230 \, \mu\text{s},
\end{align*}
where we take 
\begin{align*}
    v_f = \frac{2}{3} \cdot 3 \times 10^8 \text{ m/s} = 2 \times 10^8 \text{ m/s},
\end{align*}
assuming an index of refraction of $3/2$, and
\begin{align*}
    v_s = 3 \times 10^8 \text{ m/s}.
\end{align*}
Next, we assume $p_s$ is dominated by photon loss, the probability of successful photon state projection which is usually $p_p \approx \frac 1 2$~\cite{calsamiglia2001maximum}, the photon collection efficiency $p_c$ during entanglement swapping, and the detector efficiency $p_d$. Thus, using the optical fiber attenuation rate of 0.17 dB/km,
\begin{align*}
    p_s = p_p p_c p_d \cdot (10^{-0.1 \cdot 0.17 \text{ dB/km} \cdot 56.3/2 \text{ km}})^2 \approx 0.0248,
\end{align*}
where we assume $p_c \approx 0.5$~\cite{schupp2021interface} and $p_d \approx 0.9$~\cite{ma2024drone}.
We therefore have
\begin{align*}
    r_e \approx M \cdot 106 \text{ Hz}.
\end{align*}
In general, letting $d$ be the distance between the parties, with the above assumptions on $t_a$ and $p_s$, the effective entanglement rate is \begin{align*}
    r_e = M \cdot \frac{p_p p_c p_d 10^{-0.1 \alpha d} }{(\frac{d}{2v_f} + \frac {d}{2v_s} )} = M \cdot \frac{2 p_p p_c p_d 10^{-0.1 \alpha d} }{d(\frac{1}{v_f} + \frac {1}{v_s} )}.
\end{align*}

Now, 106 Hz does not seem fast enough for HFT. However, we need to differentiate between the \textbf{reaction time} and the \textbf{event time} in HFT. The reaction time is the time between the local observation and local trade decision, which can be on the order of microseconds or less. The event time, in contrast, is defined to be the time between events that trigger a trade decision, which can be on the order of tens of milliseconds or longer~\cite{clark2013three}. When both parties know when to implement the quantum strategy,\footnote{We mention here that unlike Bell experiments, events in HFT can be asynchronous. However, in many settings we can still pair events at different exchanges in a natural way so that they take part in the same TC problem. For example, for NYSE and NASDAQ, many of the triggers for high-frequency trades come from the exchanges in Secaucus, New Jersey~\cite{nasdaq_speed}. These triggers are analogous to questions from a referee in a nonlocal game sent to the two different players. Although the triggers can arrive at two different times (close to the difference of the speed-of-light delays from Secaucus to NYSE and NASDAQ), they can be locally recognized as belonging to the same TC problem, prompting the two servers to coordinate their trades in response. Note that the arrival of the trigger at NYSE and that of NASDAQ are often spacelike separated events. } 
the entanglement generation rate should be at least the inverse of the event time,
which is only 100 Hz or lower. Hence, we can achieve the desired generation rate to conduct coordinated HFT between NYSE and NASDAQ \emph{with a single memory, $M=1$}.\footnote{Note that for Type I implementations, we cannot store entanglement and thus entanglement needs to be available during every microsecond time window. The requirements on generation rate can be significantly relaxed for Type II implementations. } We mention that the technology for scaling to many memories has been demonstrated: see for instance~\cite{lee2022quantum,golter2022multiplexed}. 
As for reaction time, the quantity to be compared with is the local operation time (local gates and measurement). Sub-microsecond operations have already been demonstrated in Bell experiments that close the locality loophole. See for example~\cite{hensen2015loophole,shalm2015strong,giustina2015significant}.
We therefore conclude that \emph{current quantum technologies can already accommodate the speeds of HFT}. As mentioned above, note again that this is not strictly necessary to obtain a quantum advantage. 

For Type II implementations, fidelity is also an important specification. Entanglement generation involves noisy memory-photon entanglement, imperfect photon detection, as well as decoherence of the memories during photon traversal. Measurement of the memories themselves can also be noisy. We will evaluate how much the quantum advantage computed for the hedging problem in~\Cref{fig:anti_toxic} is robust to noise. For simplicity, we will assume the noise can be modeled by the entangled state going through a depolarizing channel. In realistic scenarios, a detailed physics simulation should be conducted. We define the \textbf{robustness} $\nu^*$ as how much depolarizing noise can be tolerated before the quantum advantage disappears. 

Now, the effect of depolarizing noise
\begin{align*}
    \rho \mapsto (1-\nu) \rho + \nu \pi,
\end{align*}
where $\pi$ is the maximally mixed state and $\nu$ is the magnitude of the noise, on the behavior of a quantum strategy for the hedging problem is the following:
\begin{align}
\label{eq:noisy_hedgeornot}
    p(d \vert o) \mapsto (1-\nu) p(d \vert o) + \nu \frac 1 4,
\end{align}
where we assume the quantum strategy achieves the quantum value and uses only qubits. Using higher dimensional systems could improve the robustness, but at the cost of incurring larger error rates for the physical implementation. For a derivation of these results, see~\Cref{app:noisy}. See also~\cite{perez2008unbounded} for a general study on noise tolerance.  Using~\Cref{eq:noisy_hedgeornot} and the utility array in~\Cref{eq:hon_utility}, we can directly compute the effect of depolarizing noise on the expected utility:
\begin{align}
    \bar u & \mapsto (1-\nu) \bar u + \nu \frac 1 4 \times [(1-p)^2 (1+1) + 2 p(1-p) (\beta + \beta + 1- \beta +1-\beta) + p^2 (1+1)] \nonumber \\
    & = (1-\nu) \bar u + \nu \frac 1 2.
\label{eq:noisy_utility}
\end{align}
Conveniently, this transformation does not depend on $p$ or $\beta$. 
We observe the second term in~\Cref{eq:noisy_hedgeornot} corresponds to a constant classical behavior where each party independently chooses one of the possible decisions $A,B$ with probability $\frac 1 2$ regardless of the observation, which achieves the expected utility of $\frac 1 2$. 
Furthermore, letting the NYSE server always issue ask orders first and the NASDAQ server issue the bid order first when where is no indicator and the ask order first when there is, we find that the expected utility is $1-p(1-p)$. Thus,
\begin{equation*}
    q^* \geq c^* \geq 1- p(1-p) \geq \frac 3 4 > \frac 1 2.
\end{equation*}
Hence, we can simply compute the robustness as
\begin{align*}
    \nu^* = \frac{q^* - c^*}{q^* - \frac 1 2}.
\end{align*}
We plot the robustness for the hedging problem as a function of $p$ and $\beta$ in~\Cref{fig:robustness}.
\begin{figure}
    \centering
    \begin{subfigure}[b]{0.85\textwidth}
    \centering
    \includegraphics[width = 0.95\textwidth, right]{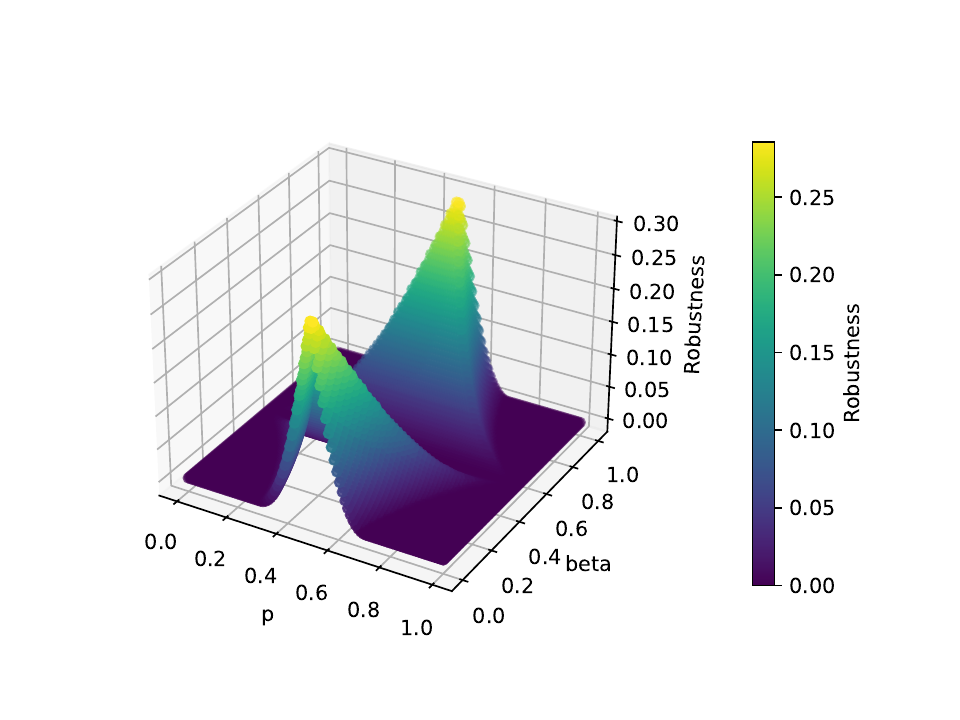}
    \caption{}
    \end{subfigure}
    \\
    \begin{subfigure}[b]{0.48\textwidth}
    \centering
    \includegraphics[width = 0.90\textwidth, left]{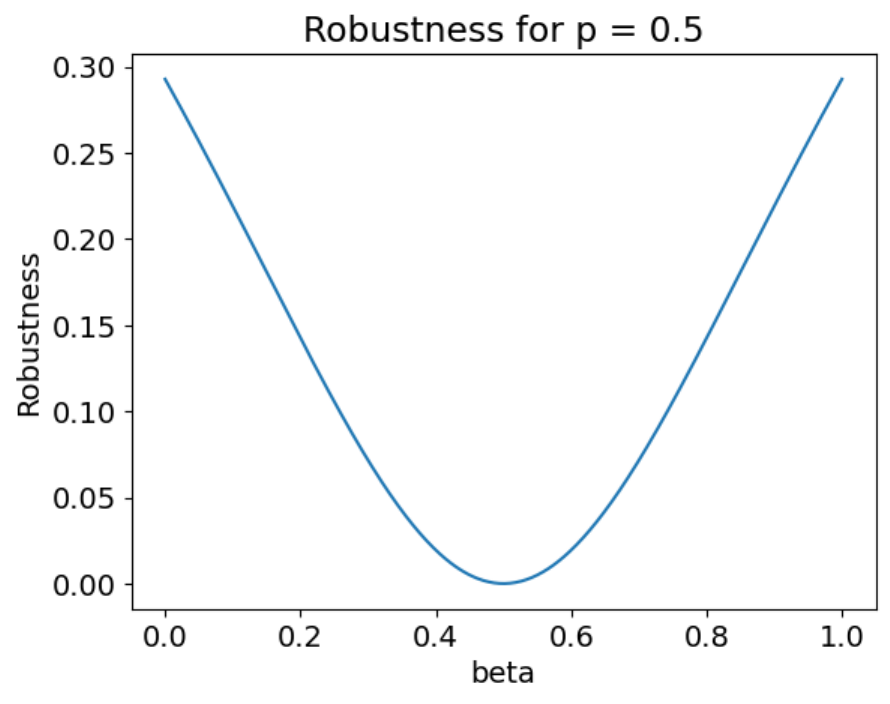}
    \caption{}
    \end{subfigure}
    \quad
    \begin{subfigure}[b]{0.48\textwidth}
    \centering
    \includegraphics[width = 0.9 \textwidth, left]{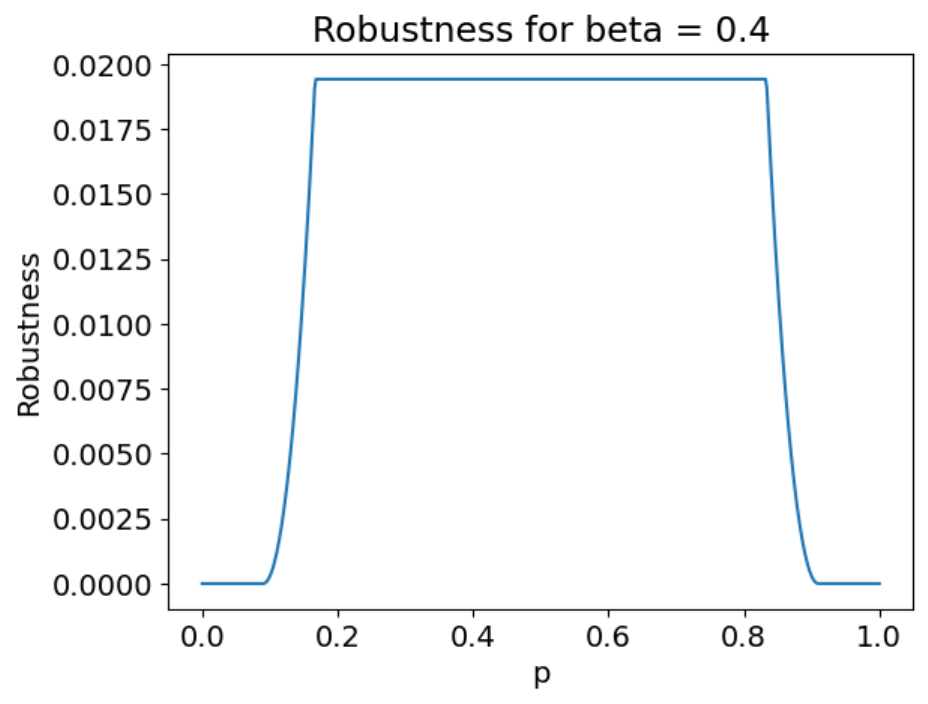}
    \caption{}
    \end{subfigure}
    \caption{(a) Robustness values computed for the hedging problem with different values of $p$ and $\beta$. (b) Robustness as a function of $\beta$ when $p=0.5$. (c) Robustness as a function of $p$ when $\beta =0.4$. }
    \label{fig:robustness}
\end{figure}
We compute robustness values of up to around 0.3. That is, optimal quantum strategies using qubits can tolerate depolarizing noise up to (not including) that magnitude and a quantum advantage would remain.
For fixed $\beta$, we also observe some regions with constant robustness, which implies in those regions the quantum value $q^*$ is a fixed affine function of the classical value $c^*$. In general, noise robustness can be interpreted as some form of normalized maximal Bell inequality violation~\cite{perez2008unbounded}. We also plot the quantum advantage for different values of $p$ and $\beta$ under various noise levels in~\Cref{fig:noisy_qadv}. As expected, the range of $p$ and $\beta$ for which there is a quantum advantage shrinks as the noise level increases.
\begin{figure}
        \centering
    \begin{subfigure}[b]{0.5\textwidth}
    \centering
    \includegraphics[width = 0.99\textwidth, right]{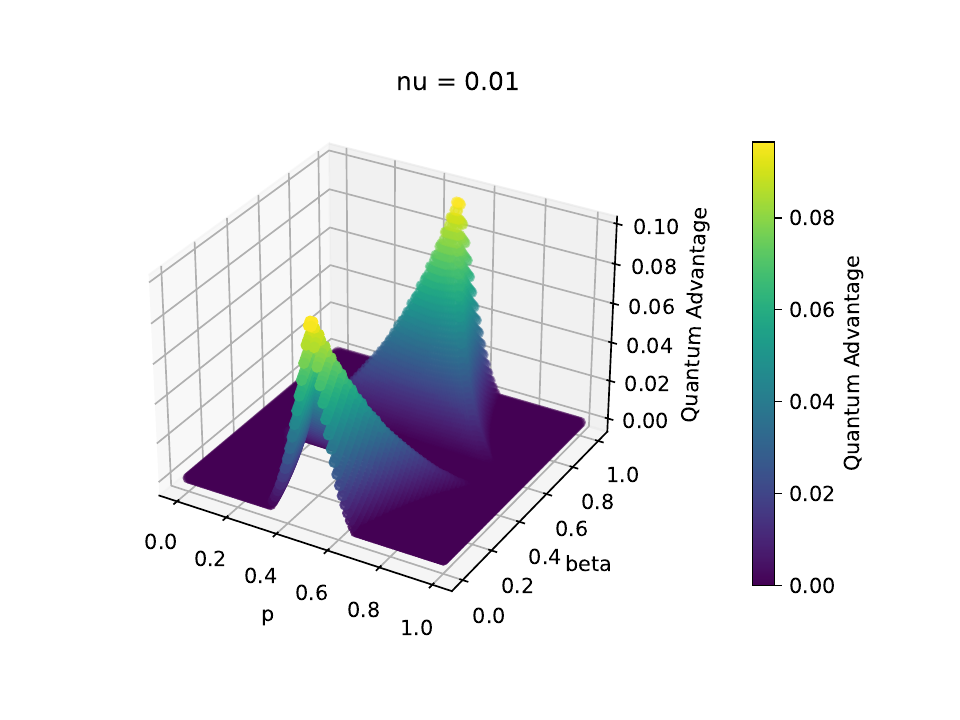}
    \caption{}
    \end{subfigure}
    \\
    \begin{subfigure}[b]{0.5\textwidth}
    \centering
    \includegraphics[width = 0.99\textwidth, left]{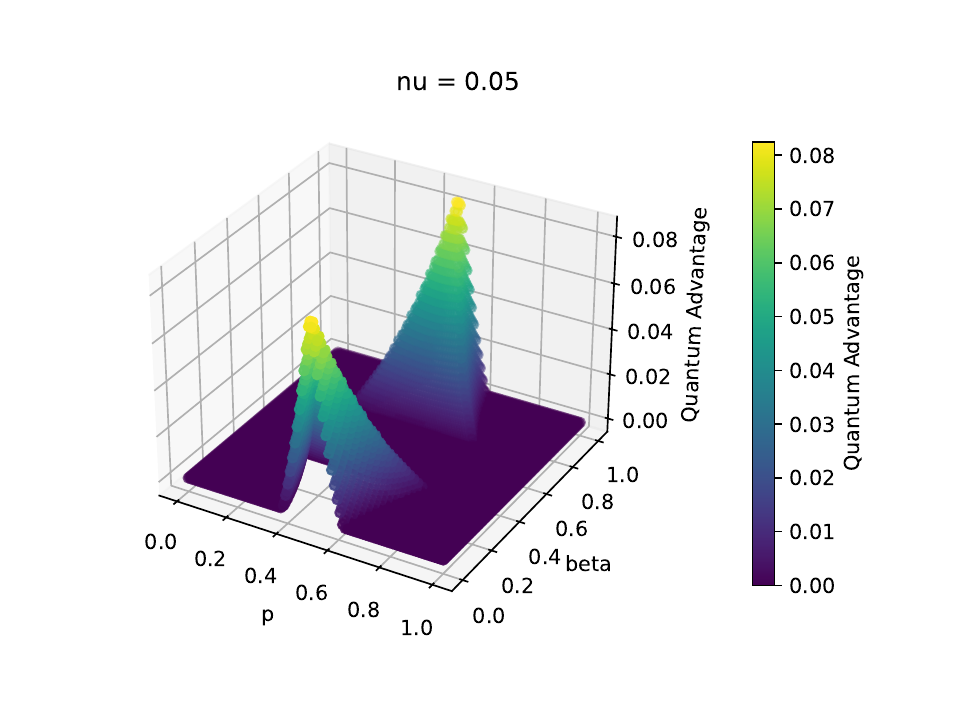}
    \caption{}
    \end{subfigure}
    \\
    \begin{subfigure}[b]{0.5\textwidth}
    \centering
    \includegraphics[width = 0.99 \textwidth, left]{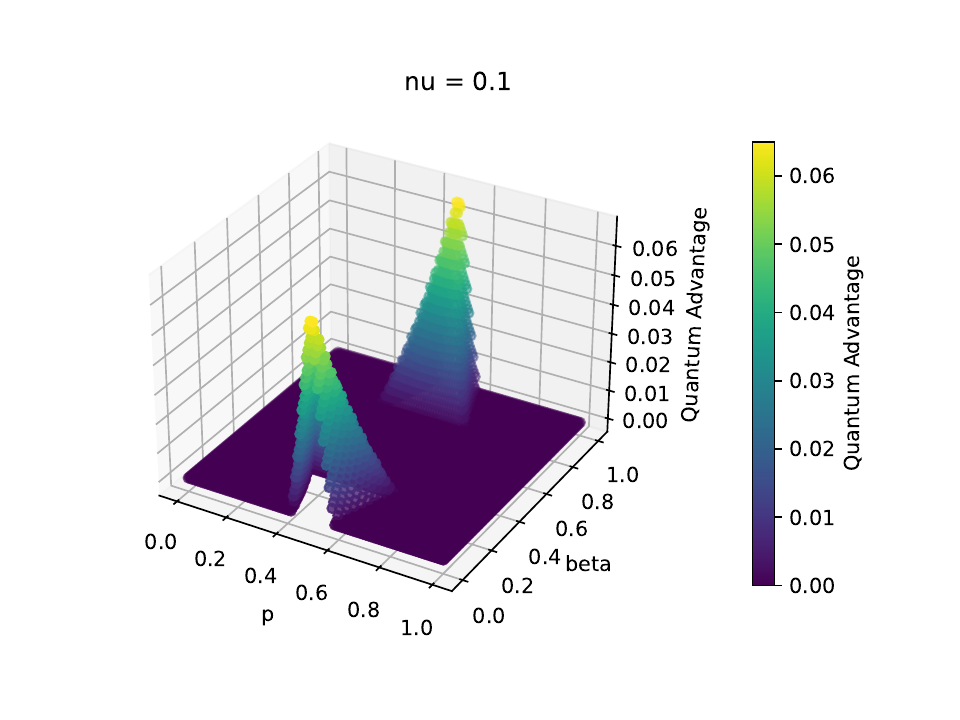}
    \caption{}
    \end{subfigure}
    \caption{Quantum advantage for the hedging problem as a function of $p$ and $\beta$ for various levels of depolarizing noise.}
    \label{fig:noisy_qadv}
\end{figure}

\section{Discussion}
\label{sec:discussion}
Quantum telepathy is a potential new application of quantum technologies. In our modern era of electronic computing where timescales are measured in microseconds or even nanoseconds, the speed of light delay is actually appreciable in many scenarios. Coordinating decisions at these timescales, which often involve space-like separated events, could benefit from using quantum entanglement. High frequency trading is a notable example. 

Aside from HFT, other possible applications include (classical) distributed computing and computer architecture. While HFT is a rather specialized field, these latter two settings are much more commonplace for anyone who uses a computer, and any interesting usage of quantum telepathy may have a greater impact. We consider the first application. Distributed computing may involve computers with a spatial separation of around 100 meters, the square root of the area of an average data center in 2024~\cite{techjury}. At this distance, any computational processes faster than 0.1 microseconds would be space-like separated. Current entangled photon sources can achieve generation rates at this timescale~\cite{yin2017satellite}, and optical fiber can achieve an efficiency of $\eta \approx 0.998$ at this distance scale.\footnote{Note in this case, the dominant source of loss would instead be from coupling. With current technologies, coupling efficiencies of 90\% or higher can be achieved~\cite{bhaskar2020experimental}. }
Hence this is the best regime to make use of current technologies. As for computer architecture, distance scales may be on the order of 1 cm or shorter in a CPU. For example, an Intel i7 CPU core package size is $50\times 25$ mm~\cite{intel}. Although it would be difficult to find any processes faster than the speed of light delay at such a distance scale, in general signal delays in a microprocessor cache can be dominated by resistive-capacitive time delays in long on-chip interconnects~\cite{nurmi2004interconnect}. Furthermore, modern computer memories can have latencies significantly longer than the speed of light delay: DDR memories can have latencies on the order of 10 ns~\cite{micron}, while SSD memories can have latencies of 10 to 100 $\mu$s~\cite{intel_ssd}. Needing to coordinate processes in different parts of the memory at these timescales may therefore engender a TC problem. However, for such timescales, we may need to use photon sources with very high generation rates such as~\cite{wakui2020ultra} to obtain an appreciable quantum advantage, as well as faster local quantum operations. Now, in the case of computer architecture, for integration purposes it is more practical to use optical waveguides to distribute photons. Compared to optical fiber, waveguides have a much higher attenuation rate of about 0.2 dB/cm~\cite{buschow2001encyclopedia}, which implies an efficiency of $\eta \approx 0.955$ for a 1 cm distance. Note that the speed of light in a waveguide is around the same as that of optical fiber: $2/3$ of that in vacuum. We propose a ``pseudo-example'' of a computer architecture scenario that can be mapped to the CHSH game in~\Cref{app:arch}.

Taking a step back, whether the application is HFT, distributed computing, or computer architecture, the crucial next step for assessing the practical usefulness of quantum telepathy is to analyze a TC problem encountered in an industrial setting with a utility array based on real data, not just a toy model such as the hedging problem. It would also be necessary to comprehensively analyze achievable efficiencies, fidelities, and entanglement generation rates to evaluate how big of a quantum advantage can be attained. We leave this for future research. Indeed, as a potentially important application of quantum technologies, we anticipate quantum telepathy will inspire new lines of theoretical and experimental research to make the technology practically viable. On the theory side, a key research direction is to figure out techniques to compute the quantum value and a quantum strategy that realizes it in an efficient way for general TC problems, or special classes of TC problems encountered in real-life scenarios. Such techniques should ideally also be able to consider photon loss. In particular, stronger results regarding the minimum number of continuous variables needed in the general-purpose optimizer outlined in~\Cref{app:numerical} would be very useful. It would also be interesting to theoretically take into account the multi-round nature~\cite{barrett2002quantum} of TC problems encountered in practice. In addition, we have only considered cooperative games so far. Identifying applications of the theory of non-cooperative multipartite games with quantum resources~\cite{auletta2021belief} in real-world scenarios where parties cannot communicate (possibly due to latency) is also a direction worth exploring.

On the experimental side, it is important to devise ways to share programmable entangled states between the different parties. 
For Type I implementations, one possible approach is for the entanglement source to use quantum computers such as atoms~\cite{firstenberg2016nonlinear} or superconducting circuits~\cite{blais2021circuit} to produce optical or microwave entangled photons with a programmable quantum state.
For Type II implementations, quantum teleportation~\cite{bennett1993teleporting} is a natural solution. It may also be useful to devise heralded entanglement schemes that produce more general entangled states and not just the maximally entangled state.
Regardless of implementation type, the authors of~\cite{szegedy2020systems,szegedy2023systems} mention it would also be particularly interesting to physically realize quantum embezzling states~\cite{van2003universal} which can be used to produce any bipartite entangled state without communication. Such an entanglement source can accommodate all possible TC problems, which is very convenient for standardization purposes. It could also be used in more complicated scenarios where the utility array is changing over time.


\paragraph{Acknowledgment} We would like to thank Haakam Aujla, Jianxin Chen, Dalu Ding, Yong Ding, Jack Gidding, Ronald Hanson, Patrick Hayden, Zhengfeng Ji, Rongzhi Jiang, Matt Meinel, Anand Natarajan, Pierre Pocreau, Francisco Horta Ferreira da Silva, Mario Szegedy, Zhaoyou Wang and Qiming Wu for helpful discussions. We thank John Gardiner and Javier Lopez for pointing out an issue with the initial HFT scenario. We also thank Yuqing Li and Xinyu Xu for helping develop the SDP and NPA algorithms that we also used in this paper. DD would like to thank God for all of His provisions. LJ acknowledges support from the Packard Foundation (2020-71479). 

\bibliographystyle{unsrt}
\bibliography{ref}

\newpage

\appendix

\section{General TC Problems: Definitions and Facts}
\label{app:tc}
In this section we present the theory of TC problems. Most of the concepts before~\Cref{app:loss} are well known in the literature on Bell inequalities (see for example~\cite{brunner2014bell}), although we do make some generalizations and slight deviations~\cite{szegedy2020systems,szegedy2023systems} to capture essential elements of a real-world TC problem and give a convenient framework to prove some results. We will make such deviations explicit. The content in~\Cref{app:loss} and~\Cref{app:noisy} also involve concepts already in the literature, but we are not aware of any work that explicitly writes down the mathematical definitions for the general, multipartite case. Also, we are not aware if previous work proved the theoretical results in these sections.

We first define a general TC problem.
\begin{dfn}
    Let $n \geq 2$ be an integer. An \textbf{n-party TC problem} is a tuple 
    $$(\{\mathcal O_i\}_{i=1}^n, \{\mathcal D_i\}_{i=1}^n, p_\mathcal{O}(o), u),$$
    where $\mathcal{O}_i, \mathcal{D}_i$ are non-empty finite\footnote{We could also consider infinite sets~\cite{aharon2013continuous,cavalcanti2007bell}. } sets, $p_\mathcal{O}(o)$ is a probability distribution over $\mathcal O := \mathcal{O}_1 \times  \mathcal{O}_2 \times \cdots \times \mathcal{O}_n$, and $u$ is a multidimensional array indexed by observations and decisions: $u_{o_1, o_2, \cdots, o_n}^{d_1, d_2, \cdots d_n} \in \mathbb{R}$, $o_i \in \mathcal O_i$ and $d_i \in \mathcal D_i$. We also define $\mathcal D := \mathcal D_1 \times \mathcal D_2 \times \cdots \times \mathcal D_n$.
\end{dfn}
\noindent $\mathcal O_i$ and $\mathcal D_i$ are the sets of \textbf{observations} and \textbf{decisions} for party $i$. $p_\mathcal{O}(o)$ is the input distribution, a probability distribution over the set of possible observation tuples across all parties. The marginal distributions of observations for each party in general may not be independent. This is an additional element that is present in nonlocal games but not general Bell expressions as presented in~\cite{brunner2014bell}. Lastly, $u$ is the \textbf{utility array}, where\footnote{For convenience here and in other places we will often use the shorthand $u_o^d:= u_{o_1, o_2,\cdots, o_n}^{d_1, d_2, \cdots d_n}$. } $u_o^d \in \mathbb{R}$ is the utility of making a combination of decisions $d \in \mathcal D$ given the combination of observations $o \in \mathcal O$. We deviate from the theory in~\cite{brunner2014bell} by proposing a \emph{multidimensional array structure} instead of a vector structure for their equivalent notion of a Bell expression. This is more natural since for one we should differentiate the observation indices from the decision indices and each party's index from each other. We also introduce the terminology that when each party has the same number of observations $m$ and the same number of decisions $\Delta$, we have an \textbf{$(n,m,\Delta)$ problem}.
\begin{dfn}
For any $i \in [n]$, where $[n]:=\{1, 2, \cdots, n\}$, let $o_i \in \mathcal O_i, d_i \in \mathcal D_i$. We define a \textbf{behavior} as a multidimensional array $p_{o_1, o_2, \cdots o_n}^{d_1, d_2, \cdots d_n}$ as conditional probabilities $p(d \vert o)$. The \textbf{expected utility} of a behavior is given by 
\begin{equation}
\label{eq:exp_utility}
    \sum_{o \in\mathcal O} p_\mathcal{O}(o) 
    \sum_{d \in \mathcal D} p_o^d u_o^d.
\end{equation} 
\end{dfn}
\noindent The conditional probability distribution $p(d\vert o)$ describes how parties make decisions given the observations. Again, we deviate from~\cite{brunner2014bell} by choosing to describe this as a multidimensional array instead of a vector. Now, in a TC problem, the parties are not in communication. Hence, their behavior has to satisfy a \textbf{no-signaling} condition, that is, for all $i \in [n]$, $o_1, o_2, \cdots, o_{i-1}, o_{i+1}, \cdots, o_n$, $d_1, d_2, \cdots, d_{i-1}, d_{i+1}, \cdots, d_n$ in their respective alphabets,
\begin{align*}
\sum_{d_i \in \mathcal D_i} p^{d_1, d_2, \cdots, d_i, \cdots d_n}_{o_1, o_2, \cdots, o_i, \cdots, o_n} 
\end{align*}
is the same for all $o_i \in \mathcal O_i$.
In words, the marginal distribution over decisions of all parties other than party $i$ is independent of the observation of party $i$. Here, our choice using a array structure for the behavior gives the no-signaling condition a more natural interpretation as a form of multidimensional array symmetry.
Note that the TC problem formalism generalizes nonlocal games that usually only have only binary outcomes of winning or losing. The formalism can capture this notion as a special case by letting $u_o^d=1$ if $(o,d)$ satisfies the winning conditions and equal $0$ otherwise. Then, the expected utility of a behavior is the probability of winning the nonlocal game.

We next define what behaviors are possible for parties with classical and quantum resources. This part is very similar to the concepts in~\cite{brunner2014bell}, so we will be brief.
\begin{dfn}
    A \textbf{deterministic strategy} is given by $\{f_i\}_{i=1}^n$ where $f_i : \mathcal O_i \to \mathcal D_i$. 
\end{dfn}
\noindent In a deterministic strategy, each player has a local function $f_i : \mathcal O_i \to \mathcal D_i$ which they use to make decision $d_i$ given observation $o_i$. The corresponding \textbf{deterministic behavior} is given by
\begin{align*}
    p_o^d := \prod_{i=1}^n \delta_{d_i, f_i(o_i)},
\end{align*}
where $o := (o_1, o_2, \cdots, o_n)$, $d := (d_1, d_2, \cdots, d_n)$, and $\delta$ is the Kronecker delta function. The expected utility of a deterministic behavior is given by
\begin{align*}
    \sum_{o \in \mathcal O} p_\mathcal{O}(o) u_o^{f(o)},
\end{align*}
where $f:\mathcal O \to \mathcal D$ is given by $f((o_1, o_2, \cdots, o_n)) := (f_1(o_1), f_2(o_2), \cdots, f_n(o_n))$. More generally, a \textbf{classical behavior} is defined by convex combinations of deterministic behaviors, which are attained by \textbf{classical strategies}: probabilistic mixtures of deterministic strategies where shared randomness can be used.
\begin{dfn}
\label{dfn:q_strategy}
    A \textbf{quantum strategy} is a tuple $(\{q_i\}_{i=1}^n, \{M_i\}_{i=1}^n, \vert \psi\rangle)$, where $q_i\geq \vert \mathcal D_i\vert$ is a positive integer, $M_i :\mathcal O_i \to S_M(q_i, \vert\mathcal D_i\vert)$, $S_M(q_i, \vert\mathcal D_i\vert)$ being the set of all projective measurements on a $q_i$-dimensional Hilbert space consisting of $\vert \mathcal D_i\vert$ projectors, and $\vert \psi \rangle$ is a $\prod_{i=1}^n q_i$-dimensional quantum state. 
\end{dfn}
\noindent In a quantum strategy, the players share a global quantum state $\vert \psi \rangle$, where each of their shares is of dimension $q_i$. They apply a projective measurement $M_i(o_i)$ on their share when they make the observation $o_i$. There are $\vert\mathcal D_i\vert$ possible measurement outcomes, which correspond to the possible decisions each party can make. We can more conveniently denote the projectors of $M_i(o_i)$ as $\Pi_i^{(d_i \vert o_i)}$, indexed by $d_i \in \mathcal D_i$. The corresponding \textbf{quantum behavior} is then given by
    \begin{align*}
        p_o^d := \langle \psi \vert \bigotimes_{i=1}^n \Pi_{i}^{(d_i\vert o_i)} \vert \psi \rangle.
    \end{align*}
We can therefore express the expected utility as
\begin{align*}
    \langle \psi \vert \sum_{o \in \mathcal O} p_\mathcal{O}(o) \sum_{d \in \mathcal D} \bigotimes_{i=1}^n \Pi_{i}^{(d_i\vert o_i)} u_o^d \vert \psi\rangle.
\end{align*}
We will call the operator 
\begin{align}
\label{eq:bell_op}
    \sum_{o \in \mathcal O} p_\mathcal{O}(o) \sum_{d \in \mathcal D} \bigotimes_{i=1}^n \Pi_{i}^{(d_i\vert o_i)} u_o^d 
\end{align}
the \textbf{Bell operator} of a quantum strategy. For the purposes of computing the quantum value it is sufficient to compute the largest eigenvalue of the Bell operator. Note in general we can also allow density matrices and POVM elements, but such strategies can always be considered to be a quantum strategy according to~\Cref{dfn:q_strategy} but in a higher dimensional space via Naimark's dilation theorem.

\subsection{Quantum Advantage}
\label{app:qadv}
Now, we want to optimize the expected utility over different strategies. We can see from~\Cref{eq:exp_utility} that for such a purpose the key object of interest is the \textbf{weighted utility array} $w$ where $w_o^d := p_\mathcal{O}(o) u_o^d$.
Indeed, we can express the expected utility of a behavior $p(d\vert o)$ as simply
\begin{align*}
    \sum_{o \in \mathcal O, d \in \mathcal D} p_o^d w_o^d.
\end{align*}
We denote the set of all possible deterministic behaviors as $\mathbf D$ and quantum behaviors as $\mathbf Q$. 
Then, we define the \textbf{classical value} and \textbf{quantum value} of $w$ as 
\begin{align*}
    c^*(w) & := \max_{p_o^d \in \mathbf D} \sum_{o \in \mathcal O, d \in \mathcal D} p_o^d w_o^d, \\
    q^*(w) &:= \max_{p_o^d \in \mathbf Q} \sum_{o \in \mathcal O, d \in \mathcal D} p_o^d w_o^d.
\end{align*}
We will call $g(w) := q^*(w) - c^*(w)$ the \textbf{gap}. 

Now, since deterministic strategies form a subset of quantum strategies, $g(w) \geq 0$. We will particularly be interested in cases for which $g(w) > 0$, in which case we call $w$ \textbf{gapped}. Otherwise, $g(w) = 0$ and we call it \textbf{gapless}. We denote the set of all gapped weighted utility arrays by $\mathcal G$. We establish some basic properties of $\mathcal G$. We in particular explore transformations of an array that preserves the gapped property.
\begin{prp}
    $\mathcal G$ is a cone, that is, it is closed under positive scalar multiplication.
\end{prp}
\begin{proof}
Let $\alpha>0$. Then, let $w \in \mathcal G$. It is clear that
\begin{align*}
    c^*(\alpha \cdot w) = \max_{p(d\vert o) \in \mathbf D} \sum_{o \in \mathcal O, d \in \mathcal D} p(d\vert o) \alpha w(o,d) = \alpha \max_{p(d \vert o) \in \mathbf D} \sum_{o \in \mathcal O, d \in \mathcal D} p(d\vert o) w(o,d) = \alpha c^*(w).
\end{align*}
Similarly, $q^*(\alpha \cdot  w) = \alpha q^*(w)$. Thus, $g(\alpha \cdot w) = \alpha g(w) >0$ and so $\alpha \cdot w \in \mathcal G$.
\end{proof}
\noindent We also define the constant array $e \in \mathcal G$ where $e_o^d = 1$. It is clear translation by a multiple $e$ also preserves $\mathcal G$:
\begin{prp}
    Let $w$ be a weighted utility array. Then, $g(w + x e) = g(w)$ for $x \in \mathbb R$. In particular, if $w \in \mathcal G$, then $w + xe \in \mathcal G$.
\end{prp}
\begin{proof}
    We have trivially
    \begin{align*}
        c^*(w+\eta e) &= c^*(w) + x \vert \mathcal O\vert \\
        q^*(w+\eta e)&= q^*(w) + x \vert \mathcal O\vert.
    \end{align*}
    Hence, $g(w+x e) = g(w)$. The second statement is immediate.
\end{proof}

\noindent Another observation is that the ordering of observations and decision is arbitrary. Hence, the following holds.
\begin{lem}
\label{lem:perm}
    The gap of a weighted utility array is invariant under permutations of observations and decisions. That is, letting $w$ be a weighted utility array, define the multidimensional array $v$ whose elements are
    \begin{align*}
        v_{o_1, o_2, \cdots, o_n}^{d_1, d_2, \cdots, d_n} := w_{\pi_1(o_1), \pi_2(o_2), \cdots, \pi_n(o_n)}^{\sigma_1(d_1), \sigma_2(d_2), \cdots, \sigma_n(d_n)},
    \end{align*}
    where $\pi_i, \sigma_i$ are permutations of $\mathcal O_i, \mathcal D_i$, respectively. Then, $g(v) = g(w)$. In particular, if $w \in \mathcal G$, then $v \in \mathcal G$.
\end{lem}
\begin{proof}
    This is immediate via simply relabeling observations and decisions according to $\pi_i$ and $\sigma_i$, respectively, for all strategies.
\end{proof}

We also state properties that \emph{do not hold} regarding $\mathcal G$. 
\begin{fac}
\label{fac:add}
    $\mathcal G$ is not closed under addition.
\end{fac}
\noindent That is, $\mathcal G$ is a cone but not a convex cone. To see this, consider the weighted utility matrix $w_\text{CHSH}$ corresponding to the CHSH game~\cite{clauser1969proposed}:
\begin{align*}
    4 (w_\text{CHSH})_{o_1, o_2}^{d_1, d_2} := (o_1 \land o_2) \oplus (d_1 \oplus d_2) \oplus 1.
\end{align*}
Here, we let $o_i, d_i \in \{0,1\}$ so that we can use logical and bitwise addition operators.  Intuitively, the two parties have to make opposite decisions when both observations are $1$ and have to make the same decision otherwise. It is well known that $w_\text{CHSH} \in \mathcal G$. Then, define the ``anti-CHSH game'' weighted utility matrix $\bar w_\text{CHSH}$ as
\begin{align*}
    4 (\bar w_\text{CHSH})_{o_1, o_2}^{d_1, d_2} := (o_1 \land o_2) \oplus (d_1 \oplus d_2).
\end{align*}
That is, it has the opposite winning condition. By a similar argument as that of the CHSH game, $\bar w_\text{CHSH} \in \mathcal G$. However, 
\begin{align*}
    w_\text{CHSH} + \bar w_\text{CHSH}= \frac 1 4 e \not \in \mathcal G.
\end{align*}
This establishes~\Cref{fac:add}.

Furthermore, counter to intuition, scaling the weighted utility by a different positive constant for each set of observations does not always preserved gappedness. This kind of scaling can be interpreted as changing the input distribution.
\begin{fac}
\label{fac:scale}
    $\exists w \in \mathcal G, \alpha_o > 0$ such that the array $w'$, defined as $(w')_o^d := \alpha_o w_o^d$, is not in $\mathcal G$.
\end{fac}
\noindent The example is again CHSH, but this time we will make use of the correlation form:
 \begin{align*}
     \frac 1 4 \langle A_0 B_0 \rangle + \frac 1 4 \langle A_0 B_1 \rangle + \frac 1 4 \langle A_1 B_0 \rangle - \frac 1 4 \langle A_1 B_1 \rangle .
 \end{align*}
\noindent The scaling we choose is a natural scenario in which the input distribution for each party is i.i.d.\ according to a Bernoulli distribution with parameter $p \in [0,1]$ instead of the uniform distribution as is usually assumed for the CHSH game. In this case, we want to instead compute
\begin{align*}
     (1-p)^2 \langle A_0 B_0 \rangle + p(1-p) \langle A_0 B_1 \rangle + p(1-p) \langle A_1 B_0 \rangle -p^2\langle A_1 B_1 \rangle .
\end{align*}
For conciseness, we will denote the correlations with $a, b, c, d \in [-1,1]$, respectively. Then, the quantum value is given by~\cite{cirel1980quantum,tsirelson1987problems,landau1988empirical,masanes2003necessary}
\begin{align}
\label{eq:chsh_scaled}
    \max_{a,b,c,d \in [-1,1]} (1-p)^2 a + p(1-p) b + p(1-p) c - p^2 d
\end{align}
subject to 
\begin{align}
\label{eq:arcsin}
    | \arcsin a + \arcsin b +\arcsin c - \arcsin d| \leq \pi
\end{align}
and its possible permutations of $a,b,c,d$.
We will compute this in full generality. 
First, we see that if~\Cref{eq:arcsin} is strictly satisfied, we can always increase $a,b,c$ and decrease $d$ until
\begin{align}
\label{eq:chsh_scaled_cons}
    \arcsin a + \arcsin b + \arcsin c - \arcsin d = \pi
\end{align}
and the expression in~\Cref{eq:chsh_scaled} can only increase. Thus, it is sufficient to consider the equality condition~\Cref{eq:chsh_scaled_cons}. We will also see that this is sufficient to satisfy all other permuted versions of~\Cref{eq:arcsin}, so the result must be the maximum.
We make use of the method of Lagrange multipliers:
\begin{align*}
    \mathcal L = (1-p)^2 a + p(1-p)b + p(1-p) c - p^2 d - \lambda (\arcsin a + \arcsin b + \arcsin c - \arcsin d - \pi).
\end{align*}
We solve for the stationary points:
\begin{align*}
    \frac{\partial \mathcal L}{\partial a} = (1-p)^2 -  \lambda \frac{1}{\sqrt{1-a^2}} = 0,
\end{align*}
so
\begin{align*}
    a = \pm \sqrt{1- \frac{\lambda^2}{(1-p)^4}}.
\end{align*}
Similarly,
\begin{align*}
    b = \pm \sqrt{1  - \frac{\lambda^2}{p^2(1-p)^2}},
    c = \pm \sqrt{1  - \frac{\lambda^2}{p^2(1-p)^2}},
    d = \pm \sqrt{1  - \frac{\lambda^2}{p^4}}.
\end{align*}
In particular, we see that $b  = \pm c$. If $b = -c$, the correlation expression becomes
\begin{align*}
    (1-p)^2 a - p^2 d \leq (1-p)^2 + p^2.
\end{align*}
The RHS can be obtained by the classical behavior where $a = 1, b = 1, c = -1, d = -1$, which also does not attain the classical value (It is clear that depending on the value of $p$, we either want to set either the first or last term in~\Cref{eq:chsh_scaled} to be negative.). Thus, WOLOG, we will set $c$ to equal $b$.
Thus, we have the simpler optimization
\begin{align*}
    \max_{a,b,d \in [-1,1]} C,
\end{align*}
where
\begin{align}
\label{eq:chsh_bern}
    C:=(1-p)^2 a + 2 p (1-p)b - p^2 d
\end{align}
and
\begin{align}
\label{eq:chsh_bern_cons}
    \arcsin a + 2 \arcsin b - \arcsin d = \pi.
\end{align}
Let
 \begin{align*}
     \sgn(x) =
     \begin{cases}
         -1 & x< 0\\
         0 & x= 0\\
         1 & x>0
     \end{cases}
 \end{align*}
 be the sign function.
 We consider the following cases
 \begin{enumerate}
     \item $\sgn(a) \sgn(d) \geq 0$: Then, taking the cosine of both sides of
     \begin{align*}
         \arcsin a -\arcsin d = \pi - 2\arcsin b,
     \end{align*}
     we obtain
     \begin{align*}
        & \cos \arcsin \sqrt{1- \frac{\lambda^2}{(1-p)^4}} \cos \arcsin \sqrt{1- \frac{\lambda^2}{p^4}} + \sqrt{1- \frac{\lambda^2}{(1-p)^4}} \sqrt{1- \frac{\lambda^2}{p^4}}  \\
        & = - \cos 2\arcsin \sqrt{1-\frac{\lambda^2}{p^2(1-p)^2}}
     \end{align*} 
     since cosine is an even function. We simplify this equation to get
     \begin{align}
     \label{eq:lambda}
         \frac{\lambda}{(1-p)^2} \frac{\lambda}{p^2} + \sqrt{1- \frac{\lambda^2}{(1-p)^4}} \sqrt{1- \frac{\lambda^2}{p^4}} = 2 (1 - \frac{\lambda^2}{p^2(1-p)^2}) - 1.
     \end{align}
     We plug this into Mathematica~\cite{mathematica} and obtain
     \begin{align*}
       \lambda = 0 , \pm \frac{1}{2\sqrt{2}} \sqrt{(2p^2-1)(2p^2-4p+1)} =: \pm \lambda^*.
     \end{align*}
     The $\lambda =0$ solutions are classically attainable solutions, which might be optimal for certain values of $p$. 
     
     \item $\sgn(a) \sgn(d) < 0$: 
     We follow similar steps to get
     \begin{align*}
       \lambda = \pm \lambda^*.
     \end{align*}
     
 \end{enumerate}
We need to check what conditions on $p \in [0,1]$ guarantee that $\lambda^* \in \mathbb{R}$. We see that this is true if
\begin{align*}
    p \in [1- \frac{1}{\sqrt{2}}, \frac{1}{\sqrt{2}}].
\end{align*}
We can easily check that $(\lambda^*)^2 \leq p^4, (1-p)^4$ and so is $\leq p^2(1-p)^2$, their geometric mean. This ensures that $a,b,d \in \mathbb{R}$, which means when $\lambda^* \in \mathbb R$, we obtain feasible solutions.

Now, we observe that to satisfy~\Cref{eq:chsh_bern_cons}, we must have at least one of $a,b \geq 0$. Furthermore, at least one of $b,d \geq 0$. We therefore analyze 5 cases. We will see that in every case where $\pm \lambda^*$ is feasible, 
\begin{align}
\label{eq:dsmall}
    \arcsin d \leq \arcsin a, \arcsin b.
\end{align}
Since $2 (\arcsin x - \arcsin y) \geq -2\pi$ for all $x,y$ and~\Cref{eq:chsh_scaled_cons} holds, all permutations of~\Cref{eq:arcsin} are satisfied.
 \begin{enumerate}
     \item $a,b,d \geq 0$: Then $\sgn(a) \sgn(d) \geq 0$. Hence, 
     \begin{align*}
     C\vert_{\lambda = \pm \lambda^*} =  \frac{1}{2\sqrt{2}}\left(|2p^2-6p+3|+ 4p^2 - 4p +2  - |2p^2+2p-1|\right),
     \end{align*}
     which is feasible when $p\in [1-\frac{1}{\sqrt{2}}, \frac{-1+\sqrt{3}}{2}]$ by checking if~\Cref{eq:chsh_bern_cons} is satisfied. Note that the three-term expression is written in the same order as~\Cref{eq:chsh_bern}. In this range for $p$, we can further simplify this to
     \begin{align*}
         C \vert_{\lambda = \pm \lambda^*} = \sqrt{2} [1-2p(1-p)].
     \end{align*}
     Note that since $p \leq \frac{-1+\sqrt{3}}{2} < 0.5$, 
     $\vert a \vert \geq \vert b\vert \geq \vert d\vert.$
     As $a,b,d$ are all non-negative,~\Cref{eq:dsmall} holds. 
     There is also the classical solution which is possible for all $p \in [0,1]$:
     \begin{align*}
         C_{\lambda = 0} = -2p^2+1.
     \end{align*}
     \item $a,b \geq 0$, $d <0$: 
     \begin{align*}
         C \vert_{\lambda = \pm \lambda^*} = \sqrt{2} [1-2p(1-p)],
     \end{align*}
     feasible when $p \in [\frac{-1+\sqrt{3}}{2}, \frac{3-\sqrt{3}}{2}]$. Due to the signs of $a,b,d$ we clearly have~\Cref{eq:dsmall}.
     \item $a,d \geq 0$, $b <0$: One can check the $\lambda^*$ solution is not feasible for any $p \in [1-\frac{1}{\sqrt{2}},\frac{1}{\sqrt{2}}]$. The classical solution in this case is
     \begin{align*}
         C\vert_{\lambda=0} = 2p^2-4p+1.
     \end{align*}
     \item $b,d\geq 0$, $a<0$: One can check the $\lambda^*$ solution is not feasible for any $p \in [1-\frac{1}{\sqrt{2}},\frac{1}{\sqrt{2}}]$.
     \item $b \geq 0$, $a,d<0$: 
     \begin{align*}
         C \vert_{\lambda = \pm \lambda^*} = \sqrt{2} [1-2p(1-p)],
     \end{align*}
     feasible when $p \in [\frac{3-\sqrt{3}}{2},\frac{1}{\sqrt{2}}]$. Since $p \geq \frac{3-\sqrt{3}}{2} > 0.5$, 
         $\vert d \vert \geq \vert b \vert \geq \vert a \vert.$
     Due to the signs of $a,b,d$, we again can conclude~\Cref{eq:dsmall}.
     The classical solution is
     \begin{align*}
         C_{\lambda=0} = -2p^2+4p-1.
     \end{align*}
\end{enumerate}
 
\noindent We combine the above results and compare the $\lambda = \pm\lambda^*$ and $\lambda =0$ solutions and find the former solutions do better for $p \in (1- \frac{1}{\sqrt{2}}, \frac{1}{\sqrt{2}})$. Outside of this range the quantum value equals the classical value, the latter of which can be easily computed. Since the CHSH inequality with an independent Bernoulli distributed input is a natural generalization of the original inequality, we summarize our result as a theorem.
\begin{thm}
    Consider the utility array $w_{\text{CHSH},p}$ for the CHSH game with each input distribution being an independent Bernoulli random variable with probability $p$. Then, the classical value is given by
    \begin{align*}
        c^* =
    \begin{cases}
        1-p^2 & 0 \leq p \leq \frac 1 2\\
        -p^2 +2p & \frac 1 2 \leq p \leq 1.
    \end{cases}
    \end{align*}
    while the quantum value is given by
    \begin{align*}
        q^* =
        \begin{cases}
            1-p^2 & 0 \leq p \leq 1 -\frac{1}{\sqrt{2}}\\
            \frac{1}{\sqrt{2}}[1-2p(1-p)] + \frac 1 2 & 1- \frac{1}{\sqrt{2}} \leq p \leq \frac{1}{\sqrt{2}}\\
            -p^2 +2p & \frac{1}{\sqrt{2}} \leq p \leq 1.
        \end{cases}
    \end{align*}
    In particular, $w_{\text{CHSH},p} \in \mathcal G$ when $p \in (1-\frac{1}{\sqrt{2}}, \frac{1}{\sqrt{2}})$ and $\not \in \mathcal G$ otherwise.
\end{thm}
\noindent Hence, for $p \not \in (1- \frac{1}{\sqrt{2}}, \frac{1}{\sqrt{2}})$, we obtain~\Cref{fac:scale}. For convenience, we also provide the numerical values for the bounds: 
\begin{align*}
        q^* =
        \begin{cases}
            1-p^2 & 0 \leq p \leq 0.293\ldots\\
            \frac{1}{\sqrt{2}}[1-2p(1-p)] + \frac 1 2 & 0.293\ldots \leq p \leq 0.707\ldots\\
            -p^2 +2p & 0.707\ldots \leq p \leq 1.
        \end{cases}
    \end{align*}
We can check that indeed for $p = \frac 1 2$,
\begin{align*}
    q^* = \frac{1}{\sqrt{2}}(1-\frac 1 2)+\frac 1 2 = \frac{1+\sqrt{2}}{2\sqrt{2}} = \cos^2(\frac \pi 8)
\end{align*}
as expected. 


\subsection{XOR Arrays}
An interesting class of TC problems have utilities that only depend on the XOR of the decisions. We make the following definition.
\begin{dfn}
    Suppose we have a TC problem where $\forall i, \mathcal D_i = \{0,1\}$. An array $m$ indexed by $o \in \mathcal O, d \in \mathcal D$ is an \textbf{XOR array} if 
    \begin{align*}
        m_o^d = f(o, \bigoplus_{i=1}^n d_i) \in \mathbb{R}.
    \end{align*}
\end{dfn}
\noindent We will call TC problems whose utility array is an XOR array an \textbf{XOR problem}.  Note that this implies the weighted utility array is also an XOR array.
Such problems are equivalent to correlation expressions~\cite{brunner2014bell} and the quantum values can be explicitly computed via a semidefinite program (SDP)~\cite{cleve2004consequences,wehner2006tsirelson}. 

An interesting observation is that we can conclude $g(\bar w_\text{CHSH}) = g(w_\text{CHSH})$ via~\Cref{lem:perm} applied to $w_\text{CHSH}$ and setting $\pi_1$ be a bit flip while letting all else be equal. We can generalize this observation by making the following definition.
\begin{dfn}
    Let $m$ be an XOR array, where $m_o^d = f(o, \bigoplus_{i=1}^n d_i)$. Then, define $\bar m$ as the \textbf{anti-array} that has elements
    \begin{align*}
        \bar m_o^d := f(o, \lnot \bigoplus_{i=1}^n d_i),
    \end{align*}
    where $\lnot x$ denotes the bit flip of $x$.
\end{dfn}
\noindent Intuitively, if $w$ is a nonlocal XOR game, then $\bar w$ is the same game with the opposite winning conditions. We can now easy establish the following
\begin{prp}
    Suppose $w$ is an XOR array. Then, $g(\bar w) =g(w)$. 
\end{prp}
\begin{proof}
    Let $\pi_1$ be the bit flip permutation on $\{0,1\}$. Then,
    \begin{align*}
        \bar w_o^d = f(o, \lnot \bigoplus_{i=1}^n d_i) = f(o, \pi_1(d_1) \oplus \bigoplus_{i=2}^n d_i) = w_o^{\pi_1(d_1), d_2, d_3, \cdots, d_n}.
    \end{align*}
    Hence, by~\Cref{lem:perm}, the conclusion follows.
\end{proof}
\noindent Intuitively, for XOR problems, the procedure of reversing the winning conditions is equivalent to local index permutation of the utility array, thereby preserving the gap. 


A central result for XOR problems with two parties is the following theorem by Tsirelson.
\begin{thm}[Tsirelson~\cite{cirel1980quantum}]
\label{thm:tsirelson}
    The following four conditions for real numbers $c_{kl}$, $k = 1,\cdots ,m$, $l = 1, \cdots, n$ are equivalent.
    \begin{enumerate}
        \item There are a C$^*$ algebra $\mathcal A$ with identity $I$, Hermitian $A_1,\cdots, A_m,B_1,\cdots,B_n \in \mathcal A$, and a state $f$ on $\mathcal A$ such that, for every $k,l$,
        \begin{align*}
            A_k B_l = B_l A_k; \quad -I \leq A_k \leq I; \quad -I \leq B_l \leq I; \quad f(A_k B_l) = c_{kl}.
        \end{align*}
        \item There are Hermitian operators $A_1, \cdots, A_m, B_1, \cdots B_n$ and a density matrix $\rho$ in a Hilbert space $\mathcal H$ such that, for every $k,l$,
        \begin{align*}
            A_k B_l = B_l A_k; \quad \mathrm{Spec}(A_k) \subseteq [-1,+1]; \quad \mathrm{Spec}(B_l)\subseteq [-1,+1]; \quad \tr[A_k B_l \rho] = c_{kl}.
        \end{align*}
        \item The same as 2 and in addition for every $k,l$; and $\mathcal H = \mathcal H_1 \otimes \mathcal H_2$, $A_k = A_k^{(1)} \otimes I_2$ and $B_l = I_1 \otimes B_l^{(2)}$ act on $\mathcal H_1$ and $\mathcal H_2$, respectively; furthermore
        \begin{align*}
         (A_k^{(1)})^2 = I_1, (B_l^{(2)})^2 = I_2, \tr[(A_k^{(1)} \otimes I_2) \rho]=0, \tr[(I_1 \otimes B_l^{(2)}) \rho]=0;
        \end{align*}
        besides that 
        \begin{align*}
            \{A_k^{(1)}, A_{k'}^{(1)}\} \propto I_1, \{B_l^{(2)}, B_{l'}^{(2)}\} \propto I_2;
        \end{align*} 
        $\mathcal H_1, \mathcal H_2$ are finite dimensional, obeying
        \begin{align*}
            \mathrm{dim}(\mathcal H_1) \leq 
            \begin{cases}
                2^{\frac m 2} & m \text{ even}\\
                2^{\frac{m+1}{2}} & m \text{ odd}
            \end{cases}
        \end{align*}
        and similarly for $\mathcal H_2$.
        \item There are unit vectors $x_1,\cdots,x_m,y_1,\cdots,y_n \in \mathbb{R}^{m+n}$ such that, for every $k,l$,
        \begin{align*}
            \langle x_k, y_l\rangle = c_{kl}.
        \end{align*}
    \end{enumerate}
\end{thm}
\noindent The set $\{c_{kl}\}_{k,l}$ is called a \textbf{quantum correlation matrix}. In terms of the TC problem formalism, we can relate a quantum correlation matrix $c_{kl}$ corresponding to the quantum behavior $p_o^d$ via
\begin{align}
\label{eq:q_corr}
    c(p_o^d)_{kl} := (p_{(k,l)}^{(0,0)} + p_{(k,l)}^{(1,1)}) - (p_{(k,l)}^{(0,1)} + p_{(k,l)}^{(1,0)}) = 2 p( \mathrm{XOR}=0 \vert (k,l)) - 1.
\end{align}
Note for XOR problems the expected utility of a quantum behavior $p_o^d$ only depends on the quantities $p(\mathrm{XOR}=0\vert (k,l))$. Furthermore, we can relate the operators $A_k^{(1)}, B_l^{(2)}$ with the measurement operators $\Pi_i^{(d_i \vert o_i)}$ by
\begin{align*}
    A_k^{(1)} = 2 \Pi_1^{(0 \vert k)} - I_1, B_l^{(2)} = 2 \Pi_2^{(0 \vert l)} - I_2.
\end{align*}

We make the following definition.
\begin{dfn}
    Define a projection operator $\Pi$ as \textbf{trivial} if $\Pi$ is the zero or identity operator. We define a quantum strategy to be \textbf{degenerate} if one party uses a trivial measurement operator. Otherwise, we call it \textbf{non-degenerate}.
\end{dfn}
\noindent Note that if a party with only two possible decisions uses a trivial measurement operator, they are locally implementing a deterministic strategy. We can thereby obtain the following corollary.
\begin{cor}
\label{cor:nondegenerate}
    Any possible quantum correlation matrix $c_{kl}$ can be realized using a non-degenerate quantum strategy involving two parties. Moreover, such a strategy can attain the quantum value of any two-party XOR problem.
\end{cor}
\begin{proof}
    Let $\mathcal S_q$ be a quantum strategy and the $c_{kl}$ be the quantum correlation matrix realized. If one of the parties uses a trivial measurement operator, WOLOG the first party, then for some $k$, $\Pi_1^{(0 \vert k)}$ is the zero or identity operator. Then, $A_k = \pm I$. But this implies $\tr[A_k \rho] \neq 0$ for any quantum state $\rho$. Thus, measurement operators obeying condition 3 of~\Cref{thm:tsirelson} must all be nontrivial measurement operators. The first conclusion follows. Since $c_{kl}$ are all that is involved in the expected utility of an XOR problem, the second conclusion follows.
\end{proof}
\noindent Such results can help reduce the search space for numerical optimizers that search over all possible quantum strategies. 


\subsection{Behaviors with Loss}
\label{app:loss}
Here we take a look from a theoretical standpoint the phenomenon of photon loss when physically implementing a quantum strategy. As mentioned in~\Cref{subsec:typeI}, such loss is common when using photons to distribute entanglement~\cite{mermin1986new}. In general, when loss occurs for a certain party, that party can fall back to a predetermined local deterministic strategy. We can rigorize this concept using the following definition.\footnote{Note that instead of defining a semiclassical strategy as a special case of a quantum strategy, we could define it as a generalization of a classical strategy, which would then allow for the choice of trivial measurement operators to be based on observations and shared randomness. This is more general than what we define but does not lead to higher expected utilities. Furthermore, ``semiclassical'' more often implies starting from quantum and then taking the classical limit. }
\begin{dfn}
    Let $S \subseteq [n]$. Then, we define an \textbf{$S$-semiclassical strategy} as a quantum strategy where the measurement operators of the parties in $S$ are all trivial for all possible observations. That is, $\Pi_i^{(d_i \vert o_i)}$ for $i \in S$ are either the zero or identity operator. 
\end{dfn}
\noindent In other words, the parties in $S$ are each implementing a local deterministic strategy. An \textbf{$S$-semiclassical behavior} is the corresponding behavior of an $S$-semiclassical strategy. Now, in general photon loss is stochastic, therefore leading to the following definition.
\begin{dfn}
    Let $\mathcal S_q$ be a quantum strategy, $\mathcal S_d$ a deterministic strategy, and $S \subseteq [n]$. Then, we use $\mathcal S_q \sqcup_S \mathcal S_d$ to denote the $S$-semiclassical strategy obtained by modifying the quantum strategy $\mathcal S_q$ so that the each party in $S$ locally implements the deterministic strategy $\mathcal S_d$. 
    
    Next, let $\eta_i \in [0,1]$ where $i \in [n]$. Then, the \textbf{$\{\eta_i\}_i$-lossy behavior} of the tuple $(\mathcal S_q, \mathcal S_d)$ is given by 
    \begin{align}
    \label{eq:lossy_behavior}
        p_{(\mathcal S_q , \mathcal S_d); \{\eta_i\}_i}(d\vert o) := \sum_{S \subseteq [n]} \prod_{i \in S} (1-\eta_i) \prod_{j \not \in S} \eta_j \cdot p_{\mathcal S_q \sqcup_S \mathcal S_d}(d \vert o),
    \end{align}
    where $p_\mathcal{S}(d\vert o)$ denotes the behavior corresponding to a strategy $\mathcal S$.

    Finally we define $q^*(\{\eta_i\}_i)$ as the \textbf{$\{\eta_i\}_i$-lossy value} of a TC problem the maximum expected utility with respect to all possible $\{\eta_i\}_i$-lossy behaviors.
\end{dfn}
\noindent It will be useful to define the \textbf{$\{\eta_i\}_i$-lossy Bell operator} for a tuple $(\mathcal S_q, \mathcal S_d)$
\begin{align}
\label{eq:lossy_bell_op}
    \sum_{S \subseteq [n]} \prod_{i \in S} (1-\eta_i) \prod_{j \not \in S} \eta_j \cdot
    \sum_{o \in \mathcal O} p_\mathcal{O}(o) \sum_{d \in \mathcal D} \bigotimes_{i=1}^n \Pi_{i}^{(d_i\vert o_i)} (\mathcal S_q \sqcup_S \mathcal S_d) u_o^d,
\end{align}
where $\Pi_i(\mathcal S_q \sqcup_S \mathcal S_d)$ are the measurement operators used in the $S$-semiclassical strategy $\mathcal S_q \sqcup_S \mathcal S_d$, which is itself a quantum strategy. The largest eigenvalue of this operator is highest expected utility for the choice of measurement operators over all possible shared quantum states.

It is straightforward to show that the set of all quantum behaviors $\mathbf Q$ is convex~\cite{pitowsky1986range}. Since a semiclassical behavior is a quantum behavior, by the definition in~\Cref{eq:lossy_behavior} we can conclude a lossy behavior also belongs to $\mathbf Q$. Thus, all possible $\{\eta_i\}_i$-lossy behaviors constitute a subset of $\mathbf Q$. We can interpret $\{\eta_i\}_i$ as a ``shrinking factor'' of $\mathbf Q$ to $\mathbf D \subseteq \mathbf Q$: when $\eta_i =1$ for all $i$, we can attain all of $\mathbf Q$, whereas when $\eta_i = 0$ for all $i$, we can only attain $\mathbf D$.



We will establish the following basic result.
\begin{prp}
\label{prp:trivial_222}
    Consider a (2,2,2) problem. Then, the behavior of a degenerate quantum strategy is a classical behavior. Moreover, the same is true for lossy behavior where the quantum strategy is degenerate.
\end{prp}
\begin{proof}
    The set of all possible classical behaviors is given by a polytope, called the \textbf{local polytope}~\cite{brunner2014bell}. By~\cite{fine1982hidden}, the inequalities that define the local polytope for the case of two parties, two observations, and two decisions are positivity conditions and permutations of the four terms in the CHSH inequality. 
    
    We will explicitly show that the behavior of a degenerate quantum strategy satisfies all possible CHSH inequalities. WOLOG, suppose the first party always outputs 0 when they make the first observation. When they make the second observation, they perform some quantum measurement $\{\Pi_A, I_A - \Pi_A\}$. Let the second party's measurement operators be given by $\{\Pi_B, I_B - \Pi_B\}$ and $\{\Pi_B', I_B - \Pi_B'\}$, in order of their observations. Then, we can compute
    \begin{align*}
        p(\text{XOR}=0\vert (0,0)) & = \mathrm{tr}[(I_A \otimes \Pi_B)\rho_{AB}]\\
        p(\text{XOR}=0\vert (0,1)) & = \mathrm{tr}[(I_A \otimes \Pi_B')\rho_{AB}]\\
        p(\text{XOR}=0\vert (1,0)) & = \mathrm{tr}[(\Pi_A \otimes \Pi_B)\rho_{AB}]+ \mathrm{tr}[((I_A-\Pi_A) \otimes (I_B -\Pi_B))\rho_{AB}] \\
        & =  2\mathrm{tr}[(\Pi_A \otimes \Pi_B)\rho_{AB}]+1-\mathrm{tr}[(I_A \otimes \Pi_B)\rho_{AB}]-\mathrm{tr}[(\Pi_A \otimes I_B)\rho_{AB}]\\
        p(\text{XOR}=0\vert (1,1)) & =  2\mathrm{tr}[(\Pi_A \otimes \Pi_B')\rho_{AB}]+1-\mathrm{tr}[(I_A \otimes \Pi_B')\rho_{AB}] -\mathrm{tr}[(\Pi_A \otimes I_B')\rho_{AB}].
    \end{align*}

    Now, the CHSH inequalities are usually expressed in terms of the quantum correlation matrix elements $c_{kl}$. We first consider the usual permutation 
    \begin{align*}
        &c_{00}+c_{01}+c_{10}-c_{11} \\
        &= 2\mathrm{tr}[(I_A \otimes \Pi_B)\rho_{AB}] -1+2\mathrm{tr}[(I_A \otimes \Pi_B')\rho_{AB}]-1+4\mathrm{tr}[(\Pi_A \otimes \Pi_B)\rho_{AB}]+2-2\mathrm{tr}[(I_A \otimes \Pi_B)\rho_{AB}]\\
        &-2\mathrm{tr}[(\Pi_A \otimes I_B)\rho_{AB}]-1-4\mathrm{tr}[(\Pi_A \otimes \Pi_B')\rho_{AB}]-2+2\mathrm{tr}[(I_A \otimes \Pi_B')\rho_{AB}]+2\mathrm{tr}[(\Pi_A \otimes I_B)\rho_{AB}]+1\\
        & = -2+4\mathrm{tr}[(I_A \otimes \Pi_B')\rho_{AB}]+4\mathrm{tr}[(\Pi_A \otimes \Pi_B)\rho_{AB}]-4\mathrm{tr}[(\Pi_A \otimes \Pi_B')\rho_{AB}]\\
        &=4\mathrm{tr}[(\Pi_A \otimes \Pi_B+ (I_A-\Pi_A) \otimes \Pi_B')\rho_{AB}]-2.
    \end{align*}
    Now, it is easy to see that the operator $\Pi_A \otimes \Pi_B+ (I_A-\Pi_A) \otimes \Pi_B'$ is an orthogonal projector. Thus, its eigenvalues are either 0 or 1. Hence,
    \begin{align*}
        \vert c_{00}+c_{01}+c_{10}-c_{11} \vert \leq 2.
    \end{align*}
    By symmetry, we only need to consider one other permutation:
    \begin{align*}
        &c_{00}+c_{10}+c_{11}-c_{01} \\
        &= 2\mathrm{tr}[(I_A \otimes \Pi_B)\rho_{AB}] -1+4\mathrm{tr}[(\Pi_A \otimes \Pi_B)\rho_{AB}]+2-2\mathrm{tr}[(I_A \otimes \Pi_B)\rho_{AB}]-2\mathrm{tr}[(\Pi_A \otimes I_B)\rho_{AB}]-1\\
        &+4\mathrm{tr}[(\Pi_A \otimes \Pi_B')\rho_{AB}]+2-2\mathrm{tr}[(I_A \otimes \Pi_B')\rho_{AB}]-2\mathrm{tr}[(\Pi_A \otimes I_B)\rho_{AB}]-1-2\mathrm{tr}[(I_A \otimes \Pi_B')\rho_{AB}]+1\\
        & = 2+4\mathrm{tr}[(\Pi_A \otimes \Pi_B)\rho_{AB}]-4\mathrm{tr}[(\Pi_A \otimes I_B)\rho_{AB}] +4\mathrm{tr}[(\Pi_A \otimes \Pi_B')\rho_{AB}]-4\mathrm{tr}[(I_A \otimes \Pi_B')\rho_{AB}]\\
        & = -4\mathrm{tr}[(\Pi_A \otimes (I_B-\Pi_B)+(I_A-\Pi_A)\otimes \Pi_B')\rho_{AB}]+2.
    \end{align*}
    Again, $\Pi_A \otimes (I_B-\Pi_B)+(I_A-\Pi_A)\otimes \Pi_B'$ is an orthogonal projector. Thus,
    \begin{align*}
        \vert c_{00}+c_{10}+c_{11}-c_{01} \vert \leq 2.
    \end{align*}
    We can conclude that the quantum behavior satisfies the inequalities of the local polytope and is therefore a classical behavior.

    Now, given a tuple $(\mathcal S_q, \mathcal S_d)$ where $\mathcal S_q$ is degenerate, the corresponding lossy behavior is a convex combination of quantum behaviors (possibly semiclassical) where each behavior results from a degenerate strategy. Hence, by the previous result, the lossy behavior is a convex combination of classical behaviors and is therefore itself classical.
\end{proof}
\noindent We next prove that 2-dimensional (qubit) quantum systems are sufficient to achieve lossy values for $(n,2,2)$ problems.
\begin{prp}
\label{prp:n22}
    Consider an $(n,2,2)$ problem. Then, the $\{\eta_i\}$-lossy value can be attained where the optimal quantum strategy $\mathcal S_q$ uses only qubit quantum systems.
\end{prp}
\begin{proof}
    Fix a deterministic strategy $\mathcal S_d$. Let $\mathcal S_q$ be a quantum strategy. 
    The results in~\cite{franz2011extremal,raeburn1989c} show that the extremal points of the set of all quantum behaviors for $(n,2,2)$ problems are realized by quantum strategies using only qubit quantum systems. Thus, in this setting every quantum behavior is a convex combination of qubit quantum behaviors. 
    We therefore have
    \begin{align*}
        p_{(\mathcal S_q , \mathcal S_d); \{\eta_i\}_i}(d\vert o) &= \sum_{S \subseteq [n]} \prod_{i \in S} (1-\eta_i) \prod_{j \not \in S} \eta_j \cdot p_{\mathcal S_d}(d_S \vert o_S) p_{\mathcal S_q}(d_{\bar S} \vert o_{\bar S}) \\
        & = \sum_{S \subseteq [n]} \prod_{i \in S} (1-\eta_i) \prod_{j \not \in S} \eta_j \cdot p_{\mathcal S_d} (d_S \vert o_S) \left(\sum_k \lambda_k p_{\mathcal S_{q_k}} (d_{\bar S} \vert o_{\bar S})\right) \\
        & = \sum_k \lambda_k \sum_{S \subseteq [n]} \prod_{i \in S} (1-\eta_i) \prod_{j \not \in S} \eta_j \cdot p_{\mathcal S_d} (d_S \vert o_S)  p_{\mathcal S_{q_k}} (d_{\bar S} \vert o_{\bar S}),
    \end{align*}
    where $\bar S$ is the complementary set of $S$, $p(d_S \vert o_S)$ is the marginal probability\footnote{Note that by the no-signaling property the marginal probability does not depend on the observations of $\bar S$. } on $S \subseteq [n]$, $\lambda_k$ is a probability vector, and $\mathcal S_{q_k}$ are qubit strategies. Thus, for fixed $\mathcal S_d$, the maximum expected utility is achieved by a qubit behavior. We can then maximize over $\mathcal S_d$ to obtain the desired conclusion.
    
\end{proof}

\subsection{Noisy Quantum Behavior}
\label{app:noisy}
We give a theoretical exposition on the effect of depolarizing noise in quantum strategies. We make the following definition.
\begin{dfn}
    Let $\mathcal S_q = (\{q_i\}_{i=1}^n, \{M_i\}_{i=1}^n, \vert \psi\rangle)$ be a quantum strategy and $\nu \in [0,1]$. Then, the \textbf{$\nu$-noisy behavior of $\mathcal S_q$} is the behavior that corresponds to the strategy $\mathcal S_q$ except the quantum state $\vert \psi\rangle$ is replaced by $(1-\nu)\vert \psi \rangle\langle \psi\vert + \nu \pi$, where $\pi$ is the maximally mixed state. That is, the $\nu$-noisy behavior is given by
    \begin{align}
    \label{eq:noisy}
        p_o^d(\nu) := \tr\left[\bigotimes_{i=1}^n \Pi_i^{(d_i \vert o_i)} ( (1-\nu)\vert \psi\rangle\langle \psi\vert + \nu \pi)\right].
    \end{align}
    The \textbf{robustness} $\nu^*$ of a quantum strategy whose expected utility is at least the classical value is the smallest $\nu$ such that the expected utility of the $\nu$-noisy behavior is equal to the classical value.
\end{dfn}
\noindent We can expand out~\Cref{eq:noisy}:
\begin{align*}
    p_o^d(\nu) = \tr \left[ \bigotimes_{i=1}^n \Pi_{i}^{(d_i\vert o_i)}  ( (1-\nu)\vert \psi\rangle\langle \psi\vert + \nu \pi)\right] = (1-\nu) p_o^d(0) + \nu \frac{\mathrm{rank}\left[\bigotimes_{i=1}^n \Pi_i^{(d_i \vert o_i)}\right]}{\prod_{i=1}^n q_i}.
\end{align*}
We see that the effect of the noise is completely determined by the ranks of the measurement operators. Now, the fraction in the second term can be interpreted as a behavior. Indeed, we observe that it factorizes into behaviors for individual parties:
\begin{align}
\label{eq:factorizable}
   \frac{\mathrm{rank}\left[\bigotimes_{i=1}^n \Pi_i^{(d_i \vert o_i)}\right]}{\prod_{i=1}^n q_i} = \prod_{i=1}^n \frac{\mathrm{rank}[\Pi_i^{(d_i \vert o_i)}]}{q_i}.
\end{align}
Hence, it is a classical behavior. We can therefore interpret $\nu$ as a shrinking factor of $\mathbf Q$ to the set of all classical behaviors (with rational probabilities) that are factorizable, that is, where the behavior of each party is independent. Note that this set includes $\mathbf D$, so we can always attain the classical value even for $\nu=1$. Furthermore, this implies for a quantum strategy with an expected utility higher than the classical value, the expected utility for the $\nu$-noisy behavior is monotonically decreasing with $\nu$.

For the case of the hedging problem when there is a quantum advantage, if a quantum strategy only uses qubits and achieves the quantum value, the ranks of the measurement operators must be rank 1 according to~\Cref{prp:trivial_222}. Hence, in this case 
\begin{align*}
    p_o^d(\nu) = (1-\nu) p_o^d(0) + \nu  \frac 1 4.
\end{align*}

We claim that in the presence of depolarizing noise, increasing the dimensions of the quantum systems used can increase the expected utility, even for the CHSH inequality. Recall the utility array is given by
\begin{align*}
    \mathcal U =
    \bordermatrix{
    & 0,0 & 0,1 & 1,0 & 1,1 \cr
    0,0 & 1 & 0 & 0 & 1 \cr
    0,1 & 1 & 0 & 0 & 1 \cr
    1,0 & 1 & 0 & 0 & 1 \cr
    1,1 & 0 & 1 & 1 & 0
    },
\end{align*}
and the inputs are uniformly distributed. Then, by~\Cref{thm:tsirelson} and~\Cref{cor:nondegenerate}, there exists an optimal qubit strategy for the noiseless setting, call it $\mathcal S_q$, which is non-degenerate. Thus, the factorizable behavior in~\Cref{eq:factorizable} must be the uniform distribution. The expected utility for the $\nu$-noisy behavior of $\mathcal S_q$ is 
\begin{align*}
    (1-\nu) \cos^2\frac \pi 8 + \nu \frac 1 2.
\end{align*}
Now, consider modified quantum strategy $\mathcal T_q$ where both parties use ququarts (dimension 4). Instead of the original entangled state $\vert \psi\rangle$, we use $\vert \psi \rangle \otimes \vert 00 \rangle$. And for the all the projectors $\Pi_i^{(0 \vert o_i)}$ for the $0$ decision, we replace it with $\Pi_i^{(0 \vert o_i)} \otimes \vert 0 \rangle \langle 0 \vert$ (the corresponding 1 decision projector is the orthogonal complement). This clearly preserves the expected utility of the noiseless quantum behavior. However, the ranks of the projectors are not doubled while the quantum dimensions are. It is not difficult to see that the factorizable behavior in this case is, in matrix form,
\begin{align*}
    \bordermatrix{
    & 0,0 & 0,1 & 1,0 & 1,1 \cr
    0,0 & \frac{1}{16} & \frac{3}{16} & \frac{3}{16} & \frac{9}{16} \cr
    0,1 & \frac{1}{16} & \frac{3}{16} & \frac{3}{16} & \frac{9}{16} \cr
    1,0 & \frac{1}{16} & \frac{3}{16} & \frac{3}{16} & \frac{9}{16} \cr
    1,1 & \frac{1}{16} & \frac{3}{16} & \frac{3}{16} & \frac{9}{16} 
    },
\end{align*}
which implies for $\mathcal T_q$, the $\nu$-noisy behavior achieves an expected utility of
\begin{align*}
    (1-\nu) \cos^2 \frac \pi 8 + \nu \frac{9}{16} > (1-\nu) \cos^2\frac \pi 8 + \nu \frac 1 2
\end{align*}
for $\nu > 0$. This establishes the claim.

\section{Numerical Optimizer}
\label{app:numerical}
We give the details of a numerical optimizer that can be used to compute classical and quantum values. 
The optimizer itself is quite straightforward but can be used for arbitrary TC problems and not just special classes such as XOR problems. In particular, it explicitly parameterizes projective measurements with $\Delta$ possible results in a $q$-dimensional Hilbert space. It then combines these different parameterizations to parameterize an arbitrary quantum strategy, in the process reducing the total number of parameters. To our knowledge, such an explicit parameterization of a quantum strategy for the most general Bell scenario with parameter reductions is not in the literature. Note that to compute the quantum value we use a black-box optimizer, which is not guaranteed to find the globally optimal solution. Hence, we are actually computing lower bounds on the quantum value.

\subsection{Classical Value}
The classical value is straightforward to compute. The optimizer proceeds via brute force by iterating over all possible deterministic strategies and tracking the largest expected utility found. Each deterministic strategy is simply a choice of decision given an observation for each party, which for a $(\{\mathcal O_i\}_{i=1}^n, \{\mathcal D_i\}_{i=1}^n, p_\mathcal{O}(o), u)$ TC problem, would simply be parameterized by  
\begin{align*}
    \sum_{i=1}^n \vert \mathcal O_i \vert
\end{align*}
discrete variables, the variables for party $i$ taking $\vert \mathcal D_i\vert$ different values. Aside from this brute force approach, there are also linear programming methods as detailed in~\cite{zukowski1999strengthening,kaszlikowski2000violations}.

\subsection{Quantum Value}
Our numerical optimizer computes the quantum value via an optimization over explicit parameterizations of projective measurements on a $q$-dimensional Hilbert space $\mathcal H_q$ with $\Delta \leq q$ possible outcomes. This parameterization has both continuous and discrete variables. 

The continuous variables parameterize an orthogonal basis of $\mathcal H_q$ modulo nonzero scalar multiplication. We can do this via~\cite{spengler2010composite}, which gives a parameterization of a unitary matrix via $q^2$ real parameters $\lambda_{m,n}$, $m, n \in [q]$:
\begin{align}
\label{eq:unitary_param}
    U = \left[ \prod_{m=1}^{q-1} \left( \prod_{n=m+1}^q \exp(i P_n \lambda_{n,m}) \exp(i \sigma_{m,n} \lambda_{m,n} \right)\right]  \cdot \left( \prod_{l=1}^q \exp(i P_l \lambda_{l,l})\right),
\end{align}
where the matrix multiplication proceeds from left to right,
\begin{align*}
    P_l := \vert l \rangle \langle l \vert,\,
    \sigma_{m,n} := -i \vert m \rangle \langle n \vert + i \vert n \rangle \langle m \vert.
\end{align*}
The range for $\lambda_{m,n}$ is $[0,2\pi]$ for $m \geq n$ and $[0, \frac \pi 2]$ for $m < n$. Since an orthogonal basis modulo nonzero scalar multiplication is basically a unitary matrix modulo multiplication by a diagonal unitary, we see from~\Cref{eq:unitary_param} that we can simply remove the rightmost term to obtain our desired parameterization. Thus, in the end we only need the $q^2-q$ parameters $\lambda_{m,n}$ where $m \neq n$, and use the parameterization
\begin{align}
\label{eq:basis_param}
    U = \left[ \prod_{m=1}^{q-1} \left( \prod_{n=m+1}^q \exp(i P_n \lambda_{n,m}) \exp(i \sigma_{m,n} \lambda_{m,n} \right)\right],
\end{align}
where the columns of $U$ are the basis vectors that we use.

The discrete variable describes how to partition the $q$-dimensional Hilbert space $\mathcal H_q$, for which we have fixed the basis via the continuous variables, into $\Delta$ possible measurement results. This is a stars and bars partition, so the variable takes
\begin{align*}
    \binom{q+\Delta-1}{\Delta-1}
\end{align*}
possible values.

We have figured out how to parameterize a projective measurement for an observation of one party. We next must combine the parameters for every observation and every party to parameterize a Bell operator in~\Cref{eq:bell_op}. Note that we do not need to parameterize the shared quantum state since we can simply compute the largest eigenvalue of the Bell operator. In fact, WOLOG we can assume each party performs the measurement in the computational basis upon the first observation, with respect to some ordering of the observations. This is because we can always conjugate the Bell operator by a tensor product unitary that rotates for each party the measurement operators used upon making the first observation to the computational basis. In sum, for a TC problem $(\{\mathcal O_i\}_{i=1}^n, \{\mathcal D_i\}_{i=1}^n, p_\mathcal{O}(o), u)$, a quantum strategy where party $i$ uses a quantum system of dimension $q_i$ has a total of
\begin{align}
\label{eq:num_var}
    \sum_{i=1}^n (q_i^2-q_i)(\vert \mathcal O_i \vert -1)
\end{align}
continuous variables and 
\begin{align}
\label{eq:num_var_disc}
    \sum_{i=1}^n \vert \mathcal O_i \vert
\end{align}
discrete variables. 

Interestingly, this is not always the minimum number of variables we need: for some cases we can further eliminate variables. In the case of $(n,2,2)$ problems, we can obtain the following result.
\begin{prp}
\label{prp:n22_var}
    Consider an $(n,2,2)$ problem. Then, it is sufficient to use $n$ continuous variables to parameterize a quantum strategy that achieves the quantum value. Moreover, this is also sufficient to parameterize a quantum strategy such that paired with the appropriate deterministic strategy, it achieves the lossy value.
\end{prp}
\begin{proof}
    By~\Cref{prp:n22}, to achieve the quantum value we can assume we are only using qubit systems. By~\Cref{eq:num_var}, we therefore seem to need $\sum_{i=1}^n 2 \cdot (2-1) = 2n$ continuous variables to parameterize a Bell operator up to local unitary equivalence. We can halve this. Let $\lambda^{(j)}_{m,n}$ denote the continuous parameters for party $j$'s projective measurement upon their second observation (recall measurement upon first observation is in the computational basis). Now, for $q=2$,~\Cref{eq:basis_param} gives
    \begin{align*}
        U = 
        \begin{pmatrix}
            \cos \lambda_{1,2}^{(j)} & \sin \lambda_{1,2}^{(j)}\\
            -e^{i \lambda_{2,1}^{(j)}} \sin \lambda_{1,2}^{(j)} & e^{i\lambda_{2,1}^{(j)}} \lambda_{1,2}^{(j)}
        \end{pmatrix}.
    \end{align*}
    Let $\vert v\rangle$ denote the first column vector. Hence, considering all possible observations and partitions, the choices of a projector for party $j$ is the following:
    \begin{align}
    \label{eq:pi_choice}
        \Pi_j \in \{\vert v\rangle \langle v \vert, I - \vert v \rangle \langle v \vert, 0, I, \vert 0 \rangle \langle 0 \vert, \vert 1\rangle \langle 1\vert \}.
    \end{align}
    Then, we observe for all possibilities, the operator
    \begin{align*}
        Z_{-\lambda_{2,1}^{(j)}} \Pi_j Z_{\lambda_{2,1}^{(j)}}
    \end{align*}
    is independent of $\lambda_{2,1}^{(j)}$. Here we used the $Z$ rotation matrix
    \begin{align*}
        Z_\lambda :=
        \begin{pmatrix}
            1 & 0 \\
            0 & e^{i\lambda}
        \end{pmatrix}.
    \end{align*}
    Thus, we see that
    \begin{align*}
        & \left(\bigotimes_{j=1}^n Z_{-\lambda_{2,1}^{(j)}} \right) \sum_{o \in \mathcal O} p_\mathcal{O}(o) \sum_{d \in \mathcal D} \bigotimes_{j=1}^n \Pi_{j}^{(d_j\vert o_j)} u_o^d  \left(\bigotimes_{j=1}^n Z_{\lambda_{2,1}^{(j)}} \right) \\
        & = \sum_{o \in \mathcal O} p_\mathcal{O}(o) \sum_{d \in \mathcal D} \left(\bigotimes_{j=1}^n Z_{-\lambda_{2,1}^{(j)}} \Pi_{j}^{(d_j\vert o_j)} Z_{-\lambda_{2,1}^{(j)}}\right) u_o^d      
    \end{align*}
    is independent of $\lambda_{2,1}^{(j)}$ for all $j$. However, this operator is local unitarily equivalent to the Bell operator. This implies that the largest eigenvalue is independent of these $n$ parameters and so the first conclusion follows.

    In the lossy case, for any quantum strategy $\mathcal S_q$ and deterministic strategy $\mathcal S_d$, every term in the convex combination of~\Cref{eq:lossy_bell_op} also satisfies the above independence property since~\Cref{eq:pi_choice} still holds. The second conclusion therefore follows. 
\end{proof}

For example, by~\Cref{prp:n22_var} a quantum strategy for a $(2,2,2)$ problem attaining the quantum value can be parameterized with only $2$ continuous variables, instead of 4 according to~\Cref{eq:num_var}. By~\Cref{eq:num_var_disc}, we seem to need $2+2 = 4$ discrete variables. However, by~\Cref{prp:trivial_222}, for a $(2,2,2)$ problem a quantum strategy whose behavior is not classical uses only nontrivial measurement operators. Thus, to compute the quantum value, the 4 discrete variables can be eliminated as we only need to consider the $1+1 =2$ partition. 
In particular, for the hedging problem, we need to only optimize over 2 continuous variables. We use our numerical optimizer to obtain~\Cref{fig:app_anti_toxic},~\Cref{fig:app_robustness}, and~\Cref{fig:app_noisy_qadv}, which are equivalent (albeit with lower resolution) to ~\Cref{fig:anti_toxic}(a),~\Cref{fig:robustness}(a) and~\Cref{fig:noisy_qadv}, respectively, in the main text which are computed using the standard SDP algorithm~\cite{cleve2004consequences}. Also,~\Cref{fig:anti_toxic}(b,c) and~\Cref{fig:robustness}(b,c) are obtained with the above numerical optimizer.
\begin{figure}[h!]
    \centering
    \includegraphics[width=0.45\linewidth]{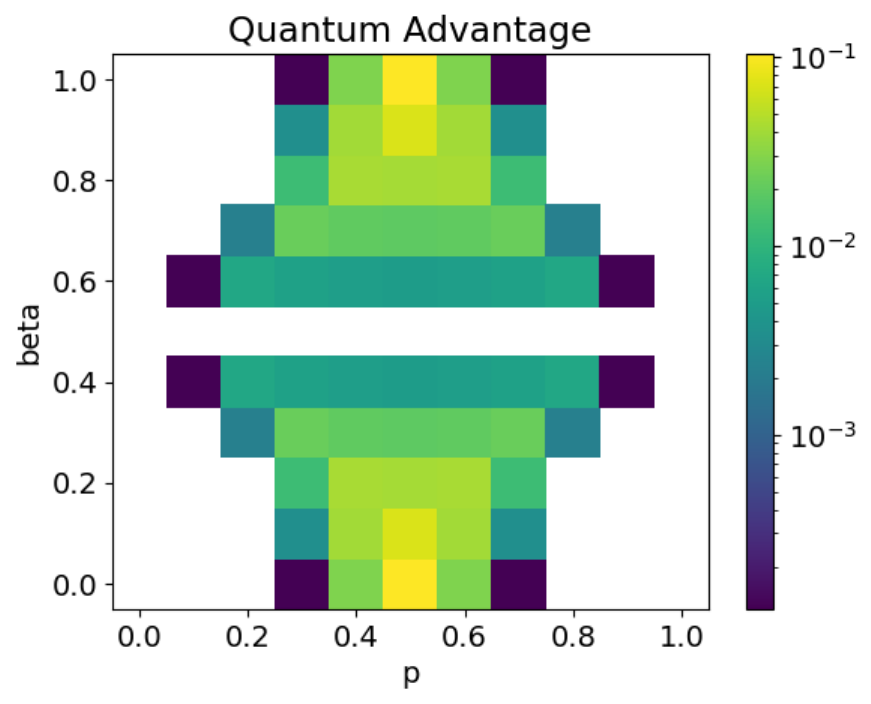}
    \caption{Quantum advantage for the hedging problem as a function of $p$ and $\beta$, where each ranges from 0 to 1 in increments of 0.1. Colors are in log scale. White squares indicate that there was no quantum advantage found up to floating point error.}
    \label{fig:app_anti_toxic}
\end{figure}
\begin{figure}[h!]
    \centering
    \includegraphics[width=0.45\linewidth]{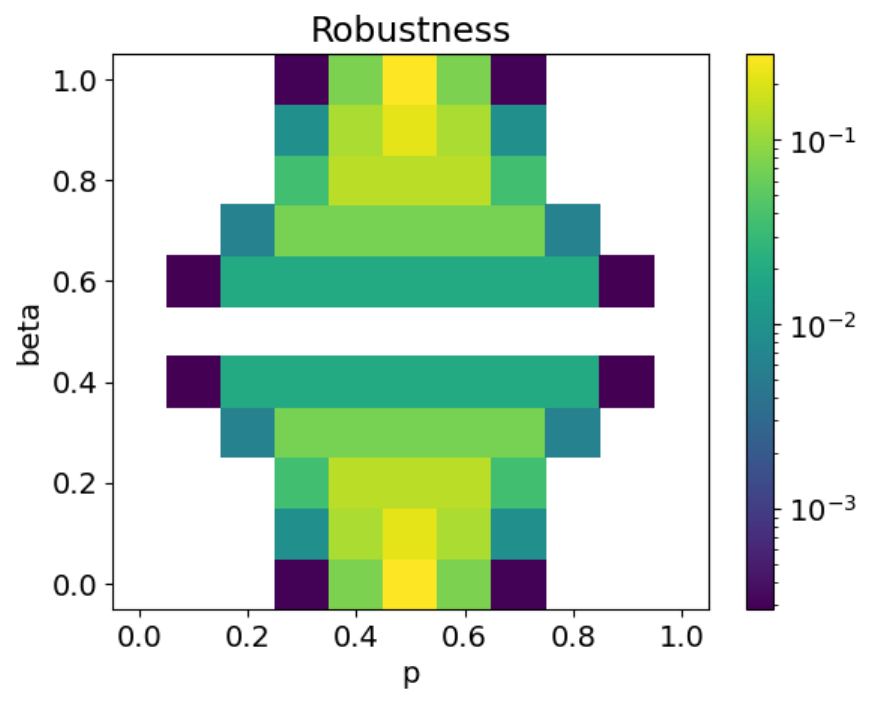}
    \caption{Robustness values computed for the hedging problem with $p$ and $\beta$ taking values between 0 and 1 with increments of 0.1. Colors are in log scale. When robustness is 0, we draw a white square.}
    \label{fig:app_robustness}
\end{figure}
\begin{figure}
        \centering
    \begin{subfigure}[b]{0.48\textwidth}
    \centering
    \includegraphics[width = 0.99\textwidth, right]{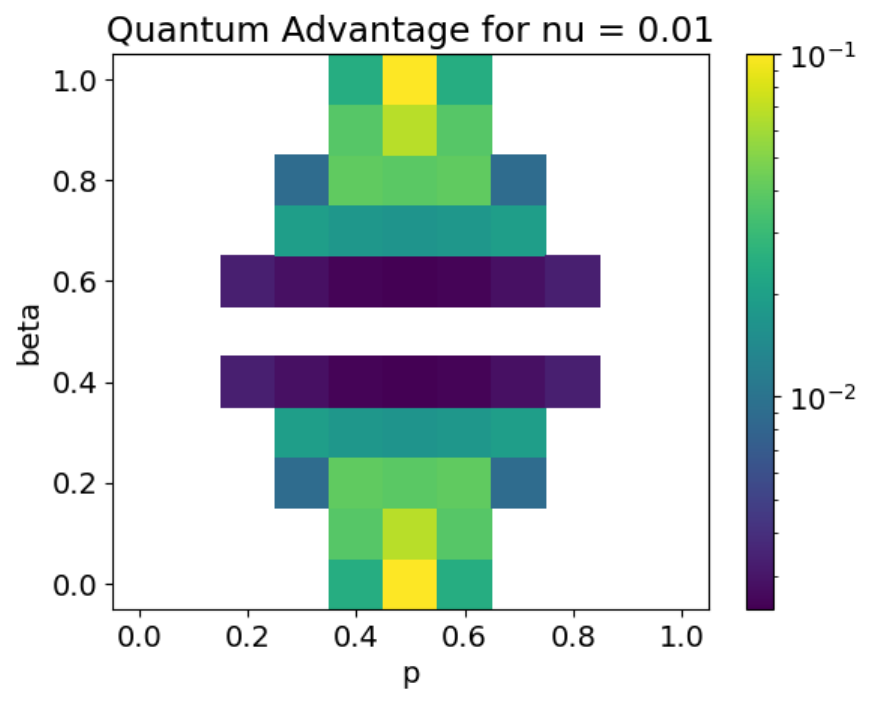}
    \caption{}
    \end{subfigure}
    \\
    \begin{subfigure}[b]{0.48\textwidth}
    \centering
    \includegraphics[width = 0.99\textwidth, left]{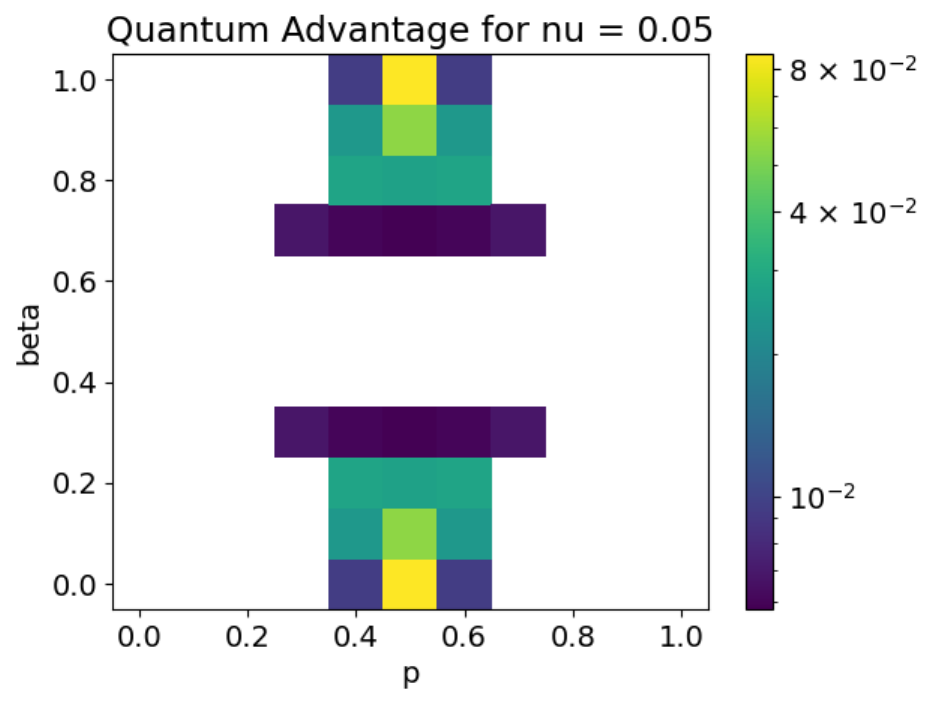}
    \caption{}
    \end{subfigure}
    \quad
    \begin{subfigure}[b]{0.48\textwidth}
    \centering
    \includegraphics[width = 0.99 \textwidth, left]{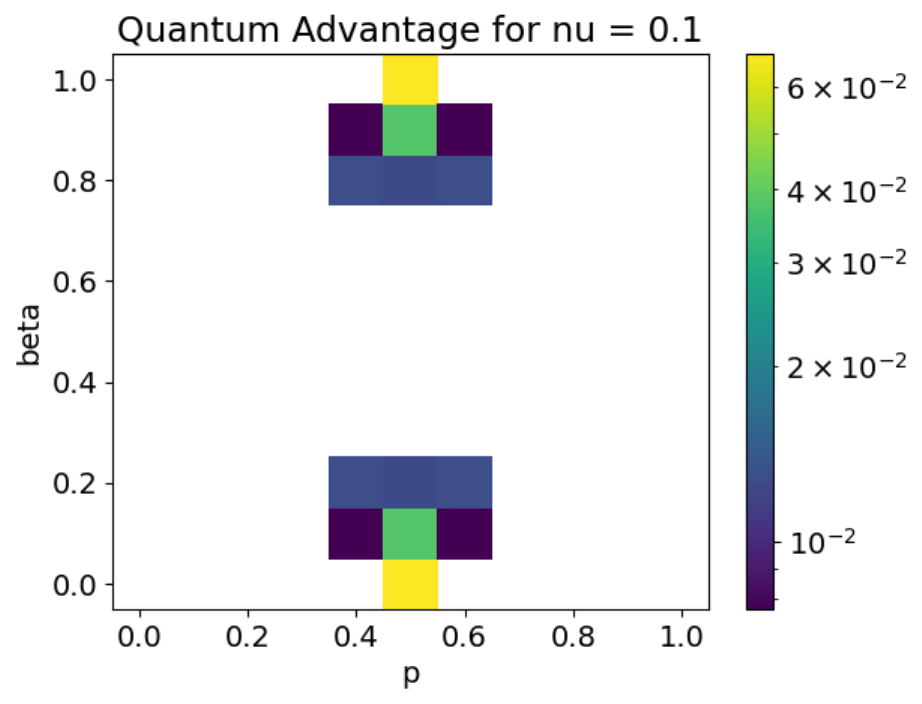}
    \caption{}
    \end{subfigure}
    \caption{Quantum advantage for the hedging problem as a function of $p$ and $\beta$, where each ranges from 0 to 1 in increments of 0.1, for various levels of depolarizing noise. Colors are in log scale. White squares indicate that there was no quantum advantage found up to floating point error.}
    \label{fig:app_noisy_qadv}
\end{figure}
\clearpage
Optimization is conducted via two methods. The first is brute force. For the continuous variables, we choose a grid of values to evaluate over and also perform a gradient descent for each point in the grid. This is implemented via \texttt{scipy.optimize}~\cite{2020SciPy-NMeth}. The discrete variables are evaluated at every possible value. This method is clearly not scalable, so we also include another method, the CMA-ES evolutionary algorithm~\cite{hansen2016cma} with discrete variables~\cite{hamano2022cma}. We use the Python implementation in~\cite{nomura2024cmaes} for our optimizer. Note that CMA-ES does not guarantee a global optimum. However, in practical scenarios it may be sufficient to find a quantum strategy with a higher expected utility than the classical value. To double-check the results in~\Cref{fig:app_anti_toxic},~\Cref{fig:app_robustness},~\Cref{fig:app_noisy_qadv},~\Cref{fig:anti_toxic}(b,c), and~\Cref{fig:robustness}(b,c), we performed both brute force search (with a grid linear size of 20 points) and CMA-ES optimizations and verified that we obtained the same results. 
\subsection{Lossy Value}
To compute the lossy value of a TC problem, we optimize over tuples $(\mathcal S_q, \mathcal S_d)$. This simply combines the parameterization of a quantum strategy with that of a deterministic strategy. The number of continuous variables is the same, but now the number of discrete variables is 
\begin{align*}
    2\sum_{i=1}^n \vert \mathcal O_i \vert,
\end{align*}
but each variable has its own corresponding range of values.
Note that we are now instead computing the largest eigenvalue of the lossy Bell operator in~\Cref{eq:lossy_bell_op}, so we again do not need to parameterize the shared quantum state. We can perform optimization via brute force where we also iterate over all values of the deterministic strategy variables. The CMA-ES optimization can be done by treating the deterministic strategy variables as just additional discrete variables.
\begin{figure}[h!]
    \centering
    \includegraphics[width=0.45\linewidth]{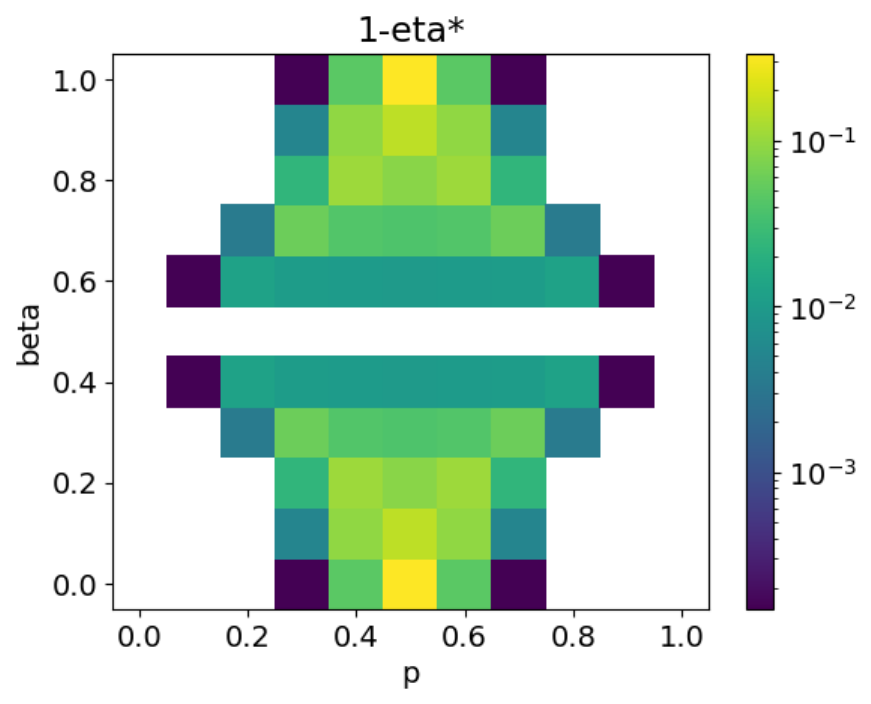}
    \caption{$1 - \eta^*$ values computed for the hedging problem with $p$ and $\beta$ taking values between 0 and 1 with increments of 0.1. Colors are in log scale. When there was no quantum advantage to begin with, by definition $1 - \eta^* = 0$ and we draw a white square.}
    \label{fig:app_eta_star}
\end{figure}

By~\Cref{prp:trivial_222} and~\Cref{prp:n22}, we can compute the lossy value of a $(2,2,2)$ problem by parameterizing a quantum strategy that only uses qubits and nontrivial measurement operators. By~\Cref{prp:n22_var}, we again only need 2 continuous variables. However, we now have 4 discrete variables parameterizing the deterministic strategy which must be included in the optimization. We perform brute force optimization over both continuous (grid linear size of 20) and discrete variables for a given $\eta$ to compute the lossy value, then perform a binary search over $\eta \in [\frac 2 3, 1]$ to find $\eta^*$, thereby obtaining the results in~\Cref{fig:app_eta_star}, which is similar (but with lower resolution) to~\Cref{fig:eta_star}(a) which was computed using a modified NPA algorithm~\cite{navascues2008convergent,tba}. We observe some small differences in computed $\eta^*$ values via NPA and our numerical optimizer primarily due to the numerical imprecision in the computed quantum value which can translate to a larger error in $\eta^*$.
We also use the above optimizer to obtain~\Cref{fig:eta_star}(b,c). 

\section{Computer Architecture Example}
\label{app:arch}
We give an example of how the CHSH game could appear in a computer architecture setting. We first emphasize that this is a ``pseudo-example'' in that it is inconsistent with how modern computer processor technologies work. Hence, the example involves a hypothetical and not realistic computing device. However, we include it to demonstrate some of the concepts and logic that could be involved in a TC problem in a realistic computer architecture scenario that admits a quantum advantage. Furthermore, it is possible that computer processor technologies could change in the future in a way that makes this example more relevant.

In this example, there are two CPU's synchronously making read accesses to memory. We assume that this is happening at such short timescales such that the CPU's do not have time to communicate with each other (This is possible due to the reasons listed in~\Cref{sec:discussion}.). We make two somewhat ad hoc assumptions:
\begin{enumerate}
    \item There are two copies of the memory. Either memory can be accessed by the CPU's.
    \item It is preferable that the CPU's read from the same memory.
\end{enumerate}
These assumptions are difficult to justify given the current technology. But continuing, assume CPU A either accesses address 1 or 2 in memory, while CPU B either accesses address 2 or 3 in memory. This setup is shown in~\Cref{fig:arch}.
\begin{figure}[h!]
    \centering
    \includegraphics[width=0.85\linewidth]{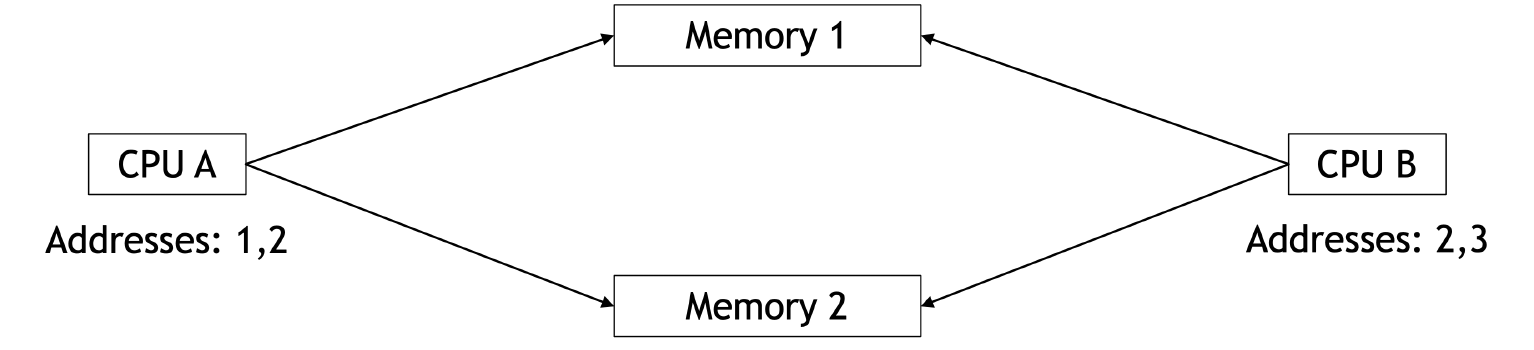}
    \caption{Two CPU's making read accesses to two copies of memory.}
    \label{fig:arch}
\end{figure}

Now, when the two CPU's want to query different addresses, by Assumption 2 they should read from the same memory. However, if they want to query the same address, they should read from different memories, as reading from the same memory would lead to an address collision. Assuming this is a simple binary game (utility matrix entries are 0 or 1), this setup therefore gives rise to the following utility matrix:
\begin{align*}
    \mathcal U =
    \bordermatrix{
    & 1,1 & 1,2 & 2,1 & 2,2 \cr
    1,2 & 1 & 0 & 0 & 1 \cr
    1,3 & 1 & 0 & 0 & 1 \cr
    2,2 & 0 & 1 & 1 & 0 \cr
    2,3 & 1 & 0 & 0 & 1
    },
\end{align*}
where the addresses to be queried by the respective CPU's label the row indices while the memories that are read label the column indices. 
Up to relabeling, this is exactly the CHSH game. Note that this TC problem can easily be generalized to involve more parties, inputs, and outputs by simply increasing the number of CPU's, addresses, and memories, respectively.

\end{document}